\documentclass[12pt,a4paper]{article} 

\usepackage[utf8]{inputenc}
\usepackage[english]{babel}
\usepackage[T1]{fontenc}
\usepackage{dirtytalk} 
\usepackage{ragged2e}
\usepackage{geometry}
\usepackage[table]{xcolor}
\definecolor{olive}{rgb}{0.3, 0.4, .1}
\definecolor{fore}{RGB}{249,242,215}
\definecolor{back}{RGB}{51,51,51}
\definecolor{title}{RGB}{255,0,90}
\definecolor{blackViolet}{RGB}{138,43,226}
\definecolor{gold}{rgb}{1.,0.84,0.}
\definecolor{JungleGreen}{cmyk}{0.99,0,0.52,0}
\definecolor{blackGreen}{cmyk}{0.85,0,0.33,0}
\definecolor{RawSienna}{cmyk}{0,0.72,1,0.45}
\definecolor{Magenta}{cmyk}{0,1,0,0}
\definecolor{wood}{RGB}{139,115,85}
\definecolor{dorange}{RGB}{255,127,0}
\definecolor{dolive}{RGB}{85,107,47}
\definecolor{drg}{RGB}{255,165,0}
\definecolor{dgreen}{rgb}{0,.5,0}
\definecolor{dblue}{rgb}{0,0,.5}
\definecolor{dred}{rgb}{0.5,0,.5}

\usepackage{graphicx}
\usepackage{rotating}
\usepackage{tikz}
\usetikzlibrary{decorations.markings}

\usepackage{gnuplottex}
\usepackage{float}

\usepackage{xspace}

\usepackage{titlesec}
\titlespacing*{\section}{0pt}{2ex}{2ex}
\titlespacing*{\subsection}{0pt}{2ex}{2ex} 
\titlespacing*{\subsubsection}{0pt}{2ex}{2ex}

\titleformat*{\section}{\large\bfseries}
\titleformat*{\subsection}{\large\bfseries}
\titleformat*{\subsubsection}{\large\bfseries}
\titleformat*{\paragraph}{\large\bfseries}
\titleformat*{\subparagraph}{\large\bfseries}

\usepackage[colorlinks,citecolor=blue,urlcolor=blue,linkcolor=blue]{hyperref}
\usepackage{cite}

\usepackage{colortbl}
\usepackage{dcolumn}
\usepackage{multirow}
\usepackage{longtable}
\usepackage{arydshln}
\usepackage{array}
\usepackage{tabularx, booktabs}
\newcolumntype{Y}{>{\centering\arraybackslash}X}

\usepackage{mathtools}
\usepackage{amsmath}
\usepackage{amsfonts}
\usepackage{amssymb} 
\usepackage{amsthm}
\usepackage{mathrsfs}
\usepackage{dsfont}
\usepackage{bm}
\usepackage{bbm}
\DeclareMathAlphabet{\mathpzc}{OT1}{pzc}{m}{it}
\newcommand\tvecpre{{\mathpalette\raiseT\top}}
\newcommand\raiseT[2]{\raisebox{-0.4ex}{$#1#2$}}
\newcommand\tvec{\mkern-2mu^\tvecpre\!}

\usepackage{stmaryrd}

\makeatletter
\newcommand\xLongLeftRightArrow[2][]{%
  \ext@arrow 0099{\LongLeftRightArrowfill@}{#1}{#2}}
\def\LongLeftRightArrowfill@{%
  \arrowfill@\Leftarrow\Relbar\Rightarrow}
\makeatother

\usepackage{simplewick}
\makeatletter
\ExplSyntaxOn
\RenewDocumentCommand{\acontraction}{ O{1ex} O{black} m m m m}{%
  \tl_if_empty:nTF{#1}{\def\wickoffset{1ex}}{\def\wickoffset{#1}}%
  {\color{#2}%
  \mathchoice
    {\acontraction@\displaystyle{#3}{#4}{#5}{#6}{\wickoffset}}%
    {\acontraction@\textstyle{#3}{#4}{#5}{#6}{\wickoffset}}%
    {\acontraction@\scriptstyle{#3}{#4}{#5}{#6}{\wickoffset}}%
    {\acontraction@\scriptscriptstyle{#3}{#4}{#5}{#6}{\wickoffset}}}}%
\RenewDocumentCommand{\bcontraction}{ O{1ex} O{black} m m m m}{%
  \tl_if_empty:nTF{#1}{\def\wickoffset{1ex}}{\def\wickoffset{#1}}%
  {\color{#2}%
  \mathchoice
    {\bcontraction@\displaystyle{#3}{#4}{#5}{#6}{\wickoffset}}%
    {\bcontraction@\textstyle{#3}{#4}{#5}{#6}{\wickoffset}}%
    {\bcontraction@\scriptstyle{#3}{#4}{#5}{#6}{\wickoffset}}%
    {\bcontraction@\scriptscriptstyle{#3}{#4}{#5}{#6}{\wickoffset}}}}%
\ExplSyntaxOff
\makeatletter

\newcommand\reallywidehat[1]{%
\savestack{\tmpbox}{\stretchto{%
  \scaleto{%
    \scalerel*[\widthof{\ensuremath{#1}}]{\kern-.6pt\bigwedge\kern-.6pt}%
    {\rule[-\textdepth/2]{1ex}{\textdepth}}
  }{\textdepth}%
}{0.5ex}}%
\stackon[1pt]{#1}{\tmpbox}%
}

\usepackage{physics} 
\usepackage{braket}

\newtheorem{conjecture}{Conjecture}[section]
\newtheorem{lemma}{Lemma}[section]
\newtheorem{corollary}{Corollary}[section]
\newtheorem{proposition}{Proposition}[section]
\newtheorem{definition}{Definition}[section]

\newtheorem{example}{Example}[section]
\newtheorem{remark}{Remark}[section]

\usepackage{algorithm}
\usepackage[noend]{algpseudocode}
\usepackage{listings}

\newcommand{\wdag}{{\vphantom{\dag}}}
\newcommand{\hatcal}[1]{\hat{\mathcal{#1}}}

\newcommand{\tcr}{\textcolor{red}}
\newcommand{\tcb}{\textcolor{blue}}

\if False
\documentclass[multi=my,crop]{standalone}

\usepackage{mathtools}
\usepackage[utf8]{inputenc}
\usepackage[french]{babel}
\usepackage[T1]{fontenc}
\usepackage{amsmath}
\usepackage{amsfonts}
\usepackage{bm}
\usepackage{stmaryrd}
\usepackage{amssymb}
\usepackage{mathrsfs}
\usepackage{bm}
\usepackage{graphicx}
\usepackage{ragged2e}
\usepackage{fancybox}
\usepackage{bbm}

\usepackage{tikz}
\usetikzlibrary{backgrounds}

\pgfdeclarelayer{bg}    
\pgfsetlayers{bg,main}
\usepackage{colortbl}
\usepackage{multimedia}
\usepackage{xcolor}
\usepackage{braket}

\usepackage{amsfonts}
\usepackage{ifthen}

\fi

\newcommand{\pdeex}[3]{
	\def\h{.5}
	\def\x{#1}
	\def\r{.5}
	\draw[line width=1pt]  
		(\x+\r, \h+\r)node[above] {$#2$}  
		arc (90:180:\r) 
		-- (\x, -\h) 
		arc (180:270:\r) node[below] {$#3$};
}

\newcommand{\pexci}[3]{
	\def\h{.5}
	\def\x{#1}
	\def\r{.5}
	\draw[line width=1pt] 
		(\x-\r, \h+\r) node[above] {$#2$}  
		arc (90:0:\r) 
		-- (\x, -\h) 
		arc (0:-90:\r) node[below] {$#3$};
}

\newcommand{\padaga}[5]{
	\def\x{#1}
	\def\e{.015}
	\def\r{.5}
	
	\ifthenelse{\equal{\detokenize{#2}}{\detokenize{occ}}}
		{ \draw[fill] (\x,-\e) arc (270:90:\e+\r) node[above] {$#3$}
			-- (\x,1+\e) arc (90:270:\r) 
			-- cycle; }
		{ \ifthenelse{\equal{\detokenize{$#2$}}{\detokenize{vir}}}	
			{ \draw[fill] (\x,\e)
				-- (\x-2*\e,\e)
				-- (\x-2*\e,\e)
				-- (\x-2*\e,-\e) arc (0:-90:2*\r) node[below] {$#3$}
				-- (\x-2*\r,-2*\r-\e) arc (-90:0:2*\r) --cycle;}
			{ \draw (\x-1,0) node[above] {$#3$};
			  \def\x{#1-1}}
		}
	
	\draw[fill] (\x,\e) 
		-- (\x+2,\e) 
		-- (\x+2,-\e) 
		-- (\x,-\e) 
		-- cycle;
	
	\ifthenelse{\equal{\detokenize{#4}}{\detokenize{occ}}}
		{ \draw[fill] (\x+2,\e) arc (-90:90:\r)  node[above] {$#5$}
			-- (\x+2,1+\e) arc (90:-90:\e+\r) 
			-- cycle; ;}
		{\ifthenelse{\equal{\detokenize{#4}}{\detokenize{vir}}}
			{\draw[fill] (\x+2,\e) 
			-- (\x+2*\e+2,\e) 
			-- (\x+2*\e+2,\e)
			-- (\x+2*\e+2,-\e) arc (180:270:2*\r) node[below] {$#5$}
			-- (\x+2*\r+2,-2*\r-\e) arc (270:180:2*\r) --cycle;}
			{\draw (\x+2,0) node[above] {$#5$};}
		}	
}

\newcommand{\adaga}[4]{
  \ifthenelse{\equal{\detokenize{#1}}{\detokenize{exci}}}
   { \draw[line width=1pt] (0,-1) arc(-90:0:1);
    \tikzset{shift={(1,0)}} }{}
  \ifthenelse{\equal{\detokenize{#1}}{\detokenize{deex}}}
   { \draw[line width=1pt] (0,1) arc(90:270:.5) -- (1,0);
    \tikzset{shift={(1,0)}} }{}

   \draw[line width=1pt] (0,0)node[above]{$#2$} -- (1.5,0)node[above]{$#4$};
   \tikzset{shift={(1.5,0)}}

  \ifthenelse{\equal{\detokenize{#3}}{\detokenize{deex}}}
   { \draw[line width=1pt] (0,0) arc(180:270:1);
    \tikzset{shift={(1,0)}} }{}
  \ifthenelse{\equal{\detokenize{#3}}{\detokenize{exci}}}
   { \draw[line width=1pt] (0,0) -- (1,0) arc(-90:90:.5);
    \tikzset{shift={(1,0)}} }{}

}

\newcommand{\exci}[2]{
   \draw[line width=1pt] 
 	     (0,-1) node[below right]{$#2$} arc (-90:0:1.5) 
 	  -- (1.5,.5) arc (0:90:.5) node[above right]{$#1$};
   \tikzset{shift={(1,0)}}
}

\newcommand{\deex}[2]{
   \draw[line width=1pt] 
   		(0,1) node[above left]{$#1$} arc (90:180:.5) 
   	  --(-.5,.5) arc (180:270:1.5)node[below left]{$#2$};
   \tikzset{shift={(1,0)}}
}

\newcommand{\exciend}[2]{
   \draw[line width=1pt] 
 	     (0,-1) node[below right]{$#2$} arc (-90:0:.5) 
 	  -- (.5,.5) arc (0:90:.5) node[above right]{$#1$};

}

\newcommand{\deexend}[2]{
   \draw[line width=1pt] 
   		(0,1) node[above left]{$#1$} arc (90:180:.5) 
   	  --(-.5,-.5) arc (180:270:.5)node[below left]{$#2$};
}

\if False
\usepackage{pgfplots}
\fi

\usetikzlibrary{calc}

\makeatletter
\newcommand\footnoteref[1]{\protected@xdef\@thefnmark{\ref{#1}}\@footnotemark}
\makeatother

\begin{document}
\title{\normalsize{\textbf{Is the TDDFRT representation of a molecular electronic transition unique?}}}

\date{} 

\maketitle

\vspace*{-2cm}

\noindent \begin{center}
\textit{by} Jérémy Morere\footnote{\label{note1}\textit{Université de Lorraine,} CNRS, LPCT, \textit{F-54000 Nancy, France}}{}$^,$\footnote{jeremy.morere@univ-lorraine.fr} \textit{and} Thibaud Etienne\footnoteref{note1}{}$^,$\footnote{thibaud.etienne@univ-lorraine.fr}\\
\end{center}

$\;$

\begin{abstract}
\noindent In this article a piece of the reference time-dependent density-functional response theory (TDDFRT) representation of molecular electronic transitions is studied through the one-body reduced difference density matrix. A first derivation route for that object has been reported in the literature based on the substitution of TDDFRT-related objects into an exact functional expression. We show in the text that the functional providing the exact object is not unique, and we question whether substituting the same TDDFRT-related objects into the alternative functional expressions leads to acceptable candidates for the approximate TDDFRT one-body reduced difference density matrix, which would question the unequivocality of the TDDFRT representation of molecular electronic transitions. We first directly address the problem using a novel diagrammatic language, inspired by Dyck languages, and specifically designed for dealing with this kind of issue. In the last part of the article we question whether it is possible or not to recast the problem in a more general framework — the equations-of-motion framework. Through all the text, the case of time-dependent Hartree-Fock is compared with that of TDDFRT to highlight what makes them so different when those questions are at stake. 
\\ $\;$ 
\end{abstract}

$\;$

\section{Introduction}

The process by which a molecular system either absorbs or emits photons is ubiquitous in nature and has been extensively studied by both theoreticians and experimentalists. When the molecule undergoes a light-induced electronic transition, its electronic cloud is reorganized. Trying to provide accurate theoretical rationalization and predictions of the optical properties of molecules is a conceptually and technically difficult task and necessitates an insightful understanding of the physical processes at stake when light-matter interaction is concerned. Those efforts have already resulted in the creation of theories, and ``methods'' for computationally approaching the ground- and excited-state properties of complex molecular systems. Those approaches may involve wave functions \cite{maurice_configuration_1995-1,david_sherrill_configuration_1999,sekino_linear_1984,koch_coupled_1990}, electron densities \cite{hirata_configuration_1999,casida_time-dependent_1995,ziegler_derivation_2014,fromager_individual_2020}, one-body reduced density matrices \cite{pernal_time-dependent_2007-1}, or Green's functions \cite{rebolini_electronic_2013,leng_gw_2018,oddershede_polarization_1978,strinati_application_1988} as central objects of interest.

In this paper we will not try to contribute to a potential answer to the ``How is the electronic structure of a molecule reorganized upon light absorption or emission?'' question. We will rather discuss the \textit{representation} of this electronic structure reorganization in our theories. Before going further, we provide here what is meant by ``representation of a molecular electronic transition'' in this paper:

\begin{definition}
The representation of a molecular electronic transition is a quadruple of sets $(\mathcal{O},\mathcal{M}, \mathcal{R},\mathcal{P})$ relating the two electronic states involved in the electronic transition, either explicitly or implicitly. $\mathcal{O}$ is a non-empty set of objects — e.g., numbers, lists, matrix representations, sets, structures; $\mathcal{M}$ is a set of maps which are such that their argument(s) can be taken in $\mathcal{O}$; the $\mathcal{R}$ set is a set of relations — e.g., {\normalfont ``$\leq$'', ``$\subseteq$'', ``$=$''.} The $\mathcal{P}$ set is a set of propositions involving at least one element of $\mathcal{O}$, and optionally elements from $\mathcal{M}$ and/or $\mathcal{R}$.
\end{definition}

\noindent Such a representation is usually associated with the theory, method, or model used for computationally approaching the electronic transition properties. Of course, the choice of the approach that is appropriate for such a description will strongly vary according to the nature of the system of interest, but there are now plenty of qualitative and quantitative tools used for analyzing the electronic structure reorganization of molecular systems during the transition between two electronic states — a task which generally consists in post-processing electronic excited-state quantum-chemical calculations.\cite{dreuw_single-reference_2005,plasser_new_2014,savarese_metrics_2017} For example, there are qualitative and quantitative studies involving the analysis of the so-called ``exciton wave function'' \cite{luzanov_interpretation_1980,furche_density_2001,tretiak_density_2002,
martin_natural_2003,batista_et_martin_natural_2004,tretiak_exciton_2005,dreuw_single-reference_2005,
mayer_using_2007,surjan_natural_2007,wu_exciton_2008,luzanov_electron_2010,li_time-dependent_2011,plasser_analysis_2012,plasser_new_2014,pluhar_visualizing_2018,plasser_newbis_2014,bappler_exciton_2014,
mewes2015communication,li_particlehole_2015,li_particle-hole_2016,plasser2015statistical,
poidevin_truncated_2016,wenzel_physical_2016,
plasser2016entanglement,savarese_metrics_2017,plasser_detailed_2017,
mai_quantitative_2018,mewes_benchmarking_2018,
skomorowski_real_2018,park_low_2018}. Other studies involve one-body reduced density matrices and one-body reduced density functions, either from a qualitative
\cite{HeadGordonJPhysChem1995,dreuw_single-reference_2005,plasser_new_2014,plasser_newbis_2014,etienne_new_2014,ronca_charge-displacement_2014,ronca_density_2014,
pastore2017unveiling} or a quantitative perspective using indicators\cite{luzanov_interpretation_1980,peach_excitation_2008,luzanov_electron_2010,
le_bahers_qualitative_2011,
garcia_evaluating_2013,GuidoJChemTheoryComput2013,etienne_toward_2014,guido_effective_2014,plasser2015statistical,
wenzel_physical_2016,poidevin_truncated_2016,plasser2016entanglement,savarese_metrics_2017,pastore2017unveiling,
plasser_detailed_2017,mai_quantitative_2018,
barca_excitation_2018,mewes_benchmarking_2018,
campetella_quantifying_2019}. Those indicators are generally scalars, or vectors in $\mathbb{R}^3$.

In this context, one object is often at the center of the analysis: the (un)relaxed one-body reduced difference density matrix — two molecular electronic states are involved in an electronic transition; the one-body reduced difference density matrix corresponding to that transition is a sum of partial traces of the difference of the two state projectors — see Ref. \cite{etienne_towards_2021}. This object will be the object of interest in this study. On the other hand, when finite-dimensional systems are concerned, \textit{time-dependent density-functional response theory} (TDDFRT \cite{casida_time-dependent_1995,casida_time-dependent_2009,dreuw_single-reference_2005,hirata_configuration_1999,ferre_density-functional_2016}) is probably the most used electronic excited-state calculation method. In this contribution we will therefore use the \textit{unrelaxed} one-body reduced difference density matrix for discussing the TDDFRT representation of molecular electronic transitions — the potential conclusions one can draw about the unrelaxed matrix actually transfer to the relaxed one.

Before going further, we need to recall that TDDFRT electronic transition properties calculations are usually done using a reference, auxiliary one-determinant state, namely the Kohn-Sham state, issued from a density-functional theory (DFT) calculation of ground-state properties. In the Kohn-Sham hypothesis, that state is not the DFT ground state but rather a wave function such that the one-body reduced density function built with it is the DFT ground-state one-body reduced density function. We also recall that TDDFRT does not come with an ansatz for the electronic excited states.

There exists a \textit{reference} expression for the TDDFRT one-body reduced difference density matrix \cite{van2000geometric,furche_adiabatic_2002,ipatov_excited-state_2009,casida_time-dependent_2009,wang2021nactddfttimedependentdensityfunctional,villalobos2023lagrangian} whose structure is identical to the one we have for the Time-Dependent Hartree-Fock (TDHF) method \cite{mclachlan_time-dependent_1964-1}. That matrix is derived using a parallel with TDHF in \cite{ipatov_excited-state_2009}, i.e., by substituting a method--specific operator to the exact state-transfer operator and by substituting the exact ground state by a one-determinant reference state in an exact double-commutator expression of the difference of expectation value of an operator between two exact electronic states.

In this paper we start by proving that the exact, reference double-commutator expression of the difference of expectation value of an operator between two exact electronic states is not unique. This fact leads us to question whether the same substitutions that lead to the reference TDDFRT one-body reduced difference density matrix can be done in the alternative expressions we bring forward. In other words, we would like to see whether the matrices we obtain with these alternative constructions are (i) different from the reference matrix, and (ii) acceptable or unacceptable as approximations of the exact one-body reduced density matrix. We believe the question is important because if more than one matrix is acceptable for approximating the exact object, then it seems to us legitimate to question whether the TDDFRT representation of a molecular electronic transition is equivocal or not. We approach our problem in three ways: First, we directly address the problem by using a diagrammatic language that has been specifically designed for that purpose since we could not address the problem in first and second quantization. The language used here is, in fact, the diagrammatic extension of a generalization of Dyck languages. We have recently introduced what we believe is a strong connection between Dyck languages and Fock-space fermionic second quantization \cite{morere2026dycklanguagefermionicsecondI}, and extended and applied this approach to the cases of (nested) commutators \cite{morere2026dycklanguagefermionicsecondII}, which are of particular interest here. We finalize our diagrammatic extension of the language in this text, and apply it for solving our problem in its original formulation — i.e., using reference double-commutator expressions and the appropriate substitutions —, and show that only the reference matrix is symmetric, which is a strict requirement here. Things become more complicated in a second step when we notice that in fact our problem can be recast in a more general formulation where the reference double-commutator expressions are all replaced by symmetrized double-commutator expressions. Finally, we try to address the problem from a different perspective by questioning whether the \textit{equations-of-motion} framework \cite{herman1981analysis,herman1979analysis,mckoy1980equations,rowe2010nuclear,rowe_equations--motion_1968,rowe1966intepretation,yeager_equations_1975,lynch_excited_nodate} allows us to deal with our issue.

\section{Background, hypotheses, and notation}

\subsection{Notation}

\noindent Through all the text, $\mathbb{K}$ denotes either the $\mathbb{C}$ or the $\mathbb{R}$ field. Bold Greek letters and bold, upper case letters from Latin alphabet will denote matrices with more than one column. For instance, whereas $\textbf{x}_{ia}$ will denote the $ia$ component of vector \textbf{x} where $ia$ is assimilated to an individual integer, $({\textbf{X}})_{i,a}$ will denote the $i\times a$ element of matrix ${\textbf{X}}$ and $(\textbf{A})_{ia,jb}$ will denote the $(ia)\times (jb)$ element of matrix $\textbf{A}$, where again $ia$ and $jb$ are assimilated to individual integers.

Through all the text, the ``$\top$'' superscript will denote the \textit{transpose} of a matrix. Zero column vectors with $n$ components will be denoted $\textbf{0}_{n\times 1}$ while $n\times n$ zero (respectively, identity) matrices will be denoted $\textbf{0}_n$ (respectively, $\textbf{I}_n$). 

\subsection{Background and hypotheses}

\subsubsection{Generalities}

In this contribution we first deal with $N$--electron molecular systems described using one-body wave functions that are normalized to unity and mutually orthogonal. Those wave functions are gathered into a set:
$$B\coloneqq \left\lbrace \phi_r \, : \, r\in\mathbb{N}\right\rbrace.$$
\noindent In the absence of external field, any $N$--electron molecular system is here assumed to be describable by a time-independent, non-degenerate, Hermitian electronic Hamiltonian denoted $\hat{H}$. The $(M+1)$ first eigenstates of $\hat{H}$ are considered, and gathered in the
$$S\coloneqq \left\lbrace \Psi_i \, : \, i\in\llbracket 0, M\rrbracket\right\rbrace$$
\noindent orthonormal set. We also define
$$S_\mathbb{K} \coloneqq \mathrm{span}_\mathbb{K}S = \left\lbrace \sum _{i\in \llbracket 0, M\rrbracket} \lambda _i \Psi _i \, : \, \forall i \in \llbracket 0, M\rrbracket, \, \lambda_i \in \mathbb{K}\right\rbrace.$$ 
\noindent We can consider the
$$C \coloneqq \left\lbrace \bigwedge_{j=1}^N \phi _{i_j} \, : \, \forall j \in \llbracket 1, N \rrbracket, \, \phi_{i_j} \in {B} \right\rbrace$$
\noindent set, and ${C}_\mathbb{K} \coloneqq \mathrm{span}_\mathbb{K} {C}$. In the definition of $C$, the ``$\bigwedge$'' symbol denotes the totally antisymmetrized tensor product. If all the elements of $S$ can be written as linear combinations of totally antisymmetrized tensor products of elements of $B$, we have
$$\forall i\in\llbracket 0, M \rrbracket, \, \Psi_i \in {C}_{\mathbb{K}}.$$
\noindent We also recall the definition of \textit{endomorphism}, that will be useful in this paper.
\begin{definition}[Endomorphism of a vector space]
Let $V$ be a vector space. Any linear transformation from $V$ to itself is called {\normalfont endomorphism}. The set of all endormorphisms of $V$ is denoted $\mathrm{End}(V)$.
\end{definition}

\subsubsection{One-body reduced difference density matrix}

The one-body reduced difference density matrix element is defined as the difference between an excited state one-body reduced density matrix and the ground state one-body reduced density matrix. 

Let $n$ be a natural integer belonging to $\llbracket 1, M \rrbracket$. An element of the one-body reduced difference density matrix in $B$ corresponding to the $\Psi_0 \rightarrow \Psi_n$ electronic transition — say the $r\times s$ element — is then defined as
\begin{equation}\label{eq:1DDM}
(\boldsymbol{\gamma}^\Delta_{0\rightarrow n})_{r,s} = \bra{\Psi_{n}}\hat{s}^\dag\hat{r}\ket{\Psi_{n}} - \bra{\Psi_{0}}\hat{s}^\dag\hat{r}\ket{\Psi_{0}},
\end{equation}
where $\hat{s}^\dag$ is the fermionic creation operator corresponding to the $\phi_s$ spin-orbital, and $\hat{r}$ is the fermionic annihilation operator corresponding to the $\phi_r$ spin-orbital.

The importance of that matrix is readily seen when considering the difference of expectation value of any symmetric one-body operator $\hat{O}_1$ — assumed here to belong to $\mathrm{End}(S_\mathbb{K})$ — relatively to $\Psi_n$ and $\Psi_0$, which reads
$$\bra{\Psi_{n}}\hat{O}_1\ket{\Psi_{n}} - \bra{\Psi_{0}}\hat{O}_1\ket{\Psi_{0}} = \mathrm{tr}(\boldsymbol{\gamma}^\Delta_{0\rightarrow n}\textbf{O}_1),$$
where $\textbf{O}_1$ is the matrix representation of $\hat{O}_1$ in ${B}$, i.e.,
$$\textbf{O}_1 \coloneqq \mathcal{M}(\hat{O}_1,{B}).$$
Using the exact $\Psi_0 \rightarrow \Psi_n$ state-transfer operator, i.e.,
\begin{equation}\label{eq:def:exactSTO}
{\hat{T}}_{0\rightarrow n} \coloneqq \ket{\Psi_n}\bra{\Psi_0},
\end{equation}
the $r\times s$ element of the one-body reduced difference density matrix can be re-expressed as
\begin{equation}\label{eq:EOM_1DDM}
(\boldsymbol{\gamma}^\Delta_{0\rightarrow n})_{r,s} 
=
\left\langle \Psi_{0} \left| \left[ {\hat{T}}^\dag_{0\rightarrow n},\left[ \hat{s}^\dag\hat{r} , {\hat{T}}_{0\rightarrow n} \right] \right] \right| \Psi_{0} \right\rangle.
\end{equation}

\subsubsection{An important remark}

\noindent When dealing with approximate excited-state calculation methods we will actually work with a finite number of real-valued spin-orbitals that are normalized to unity and mutually orthogonal:
$$\mathcal{B} \coloneqq (\varphi_r)_{r\in\llbracket 1, L \rrbracket},$$
where $L$ is a natural integer strictly greater than $N$. We will also have to consider the
$$\mathcal{C} \coloneqq \left\lbrace \bigwedge_{j=1}^N \varphi _{i_j} \, : \, \forall j \in \llbracket 1, N \rrbracket, \, \varphi_{i_j} \in \mathcal{B} \right\rbrace$$
set, and its real span, i.e., $\mathcal{C}_\mathbb{R} \coloneqq \mathrm{span}_\mathbb{R} \mathcal{C}$.

In what follows we will need to compare propositions involving some expressions such as ``$\bra{\Psi_{n}}\hat{s}^\dag\hat{r}\ket{\Psi_{n}}$'', with $r$ and $s$ pointing at spin-orbitals from $B$, with propositions involving expressions in which the fermionic second quantization operators written using the same symbols actually do point to elements of the finite-dimensional $\mathcal{B}$. The two situations will be easily distinguished because we will use upper-cased ``$\Psi$'' symbols for the exact case and lower-cased ``$\psi$'' symbols for the approached case. Similarly, we will use italic letters for basis, operators and functionals in the exact case — e.g., ``$B\,$'', ``$\hat{O}\,$'', ``$\hat{T}\,$'', ``$F\,$'', and stylized letters for basis, operators and functionals in the approached case — e.g., ``$\mathcal{B}\,$'', ``$\hat{\mathcal{O}}\,$'', ``$\hatcal{T}\,$'', ``$\mathcal{F}\,$''.

\subsubsection{TDHF and TDDFRT response equation}

In what precedes we considered electronic transitions from the ground state to \textit{definite} arrival states. For the sake of readability of the formulas, when working out expressions in the TDHF or the TDDFRT framework, we will rather consider \textit{any} electronic transition — i.e., an electronic transition from the ground state to a non-definite arrival state. The central equation — \textit{vide infra} — for these methods admit some vector solutions,
\begin{eqnarray*}
\left(
\begin{array}{c}
\textbf{x}\\\textbf{y}
\end{array}
\right),
\end{eqnarray*}
that we assume have real-valued components, with $\textbf{x}\in \mathbb{R}^{N\times(L-N)}$ and $\textbf{y}\in \mathbb{R}^{N\times(L-N)}$. The \textbf{x} and \textbf{y} vectors satisfy
$$\textbf{x}\tvec\textbf{x}-\textbf{y}\tvec\textbf{y}=1,$$
and their components are recast into two real $(N\times (L-N))$ rectangular matrices:
\begin{align*}
\forall (i,a) \in \llbracket 1, N \rrbracket \times \llbracket (N+1),L\rrbracket, \, \textbf{x}_{ia} & = (\textbf{X})_{i,a-N},\\
  \textbf{y}_{ia} &= (\textbf{Y})_{i,a-N}. 
\end{align*}
Deriving the reference one-body reduced difference density matrix for the TDHF or TDDFRT methods has been reported\cite{ipatov_excited-state_2009} to consist in replacing in \eqref{eq:EOM_1DDM} the ground state, i.e., $\Psi_0$, by the
\begin{equation}\label{eq:def:psi0}
\psi_\star \coloneqq \bigwedge_{j=1}^N \varphi_j
\end{equation}
one-determinant $N$--body reference state, and by replacing ${\hat{T}}_{0\rightarrow n}$ by ${\hat{\mathcal{T}}}$, an endomorphism on $\mathcal{C}_\mathbb{R}$ defined as
\begin{equation}\label{eq:def:Tcurved}
\mathcal{\hat{T}} \coloneqq
\sum_{i=1}^{N}\sum_{a=N+1}^{L}
 (\textbf{x}_{ia} \hat{a}^\dag\hat{i}^\wdag
-\textbf{y}_{ia} \hat{i}^\dag\hat{a}^\wdag).
\end{equation}
Doing so, one obtains what we will call the \textit{reference} one-body reduced difference density matrix, $\tilde{\boldsymbol{\gamma}}^\Delta$, i.e.,
\begin{equation}\label{eq:RPA_1DDM}
\tilde{\boldsymbol{\gamma}}^\Delta =
\left(-\textbf{X}\textbf{X}^\top - \textbf{Y}\textbf{Y}^\top\right) 
\oplus 
\left( \textbf{X}^\top\textbf{X} + \textbf{Y}^\top\textbf{Y}\right).
\end{equation}
Complete and detailed derivation of this result is given in \cite{etienne_comprehensive_2021}. The $\tilde{\boldsymbol{\gamma}}^\Delta$ matrix has two required properties for being physically sound: It is symmetric, and its trace is equal to zero. 

We notice that the elements of $\mathcal{B}$ are actually the Hartree-Fock (respectively, Kohn-Sham) orbitals when dealing with TDHF (respectively, TDDFRT). The elements of $\mathcal{B}$ have actually been ordered such that the Hartree-Fock (respectively, Kohn-Sham) $N$--body wave function is the antisymmetric tensor product of the $N$ first entries of $\mathcal{B}$ — see \eqref{eq:def:psi0}. It is because of the structure of $\mathcal{B}$ that $\tilde{\boldsymbol{\gamma}}^\Delta$ has this block-diagonal structure, with the North-West — sometimes called \textit{occupied-occupied} — block and the South-East — sometimes called \textit{virtual-virtual} — block being non-zero. The North-West block is $(N\times N)$--dimensional, while the South-East block is $((L-N)\times (L-N))$--dimensional.

\section{Statement of the problem}
\begin{definition}[Exact elementary functional]\label{def:FnA}
Let $n$ a natural integer belonging to $\llbracket 1,M\rrbracket$. The {\normalfont exact elementary functional} corresponding to the $\Psi_0\longrightarrow \Psi_n$ transition is the linear form, denoted $F_n$, and defined as
\begin{align*}
\mathrm{End}(S_\mathbb{K})&\longrightarrow\mathbb{K}\\
\hat{A}&\longmapsto F_n[\hat{A}]\coloneqq \braket{\Psi_0 |[\hat{T}_{0\rightarrow n}^\dag,[\hat{A},\hat{T}_{0\rightarrow n}]]|\Psi_0}.
\end{align*}
\end{definition}
\begin{definition}[Approached elementary functional]\label{def:FA}
The {\normalfont approached elementary functional} is the linear form, denoted $\mathcal{F}$, and defined as
\begin{align*}
\mathrm{End}(\mathcal{C}_\mathbb{R}) &\longrightarrow \mathbb{R}\\
\hatcal{A}&\longmapsto\mathcal{F}[\hatcal{A}]\coloneqq \braket{\psi_\star |[\hatcal{T}^\dag,[\hatcal{A},\hatcal{T}]]|\psi_\star}.
\end{align*} 
\end{definition}
\noindent According to Definition \ref{def:FA} and the construction done in the preceding paragraph, we have that the \textit{reference} one-body reduced difference density matrix for the TDHF and the TDDFRT methods is actually such that
\begin{align}\label{eq:EOMisFA}
\forall (r,s)\in\llbracket 1, L \rrbracket^2,\, (\tilde{\boldsymbol{\gamma}}^\Delta)_{r,s} = \mathcal{F}[\hat{s}^\dag\hat{r}],
\end{align}
\begin{proposition}\label{prop:recu_EOM_simple}
Let $n$ be a natural integer belonging to $\llbracket 1, M\rrbracket$. For every $(r,s)$ couple of integers — both belonging to $\mathbb{N}^2$ — there exists an infinite sequence of operators, $(\hat{O}_k^\wdag)_{k\in\mathbb{N}}^\wdag$, defined as
\begin{equation*}
\hat{O}_0 := \hat{s}^\dag\hat{r} \;\, \mathrm{and} \;\, \forall k\in\mathbb{N},\, \hat{O}_{k+1} \coloneqq -\dfrac{1}{2} \left[ \hat{{T}}_{0\rightarrow n}^\dag, \left[\hat{O}_k,  \hat{{T}}_{0\rightarrow n} \right] \right]
\end{equation*}
such that
$$\forall k \in \mathbb{N},\, (\boldsymbol{\gamma}^\Delta_{0\rightarrow n})_{r,s} = F_n[\hat{O}_k].$$
\end{proposition}
\begin{proof}
\noindent According to Definition \ref{def:FnA}, and to formulas \eqref{eq:1DDM} and \eqref{eq:EOM_1DDM}, we have
\begin{align*}
\forall n \in \llbracket 1, M \rrbracket, \, \forall (r,s)\in\mathbb{N}^2,\, (\boldsymbol{\gamma}^\Delta_{0\rightarrow n})_{r,s} &= \bra{\Psi_{n}}\hat{s}^\dag\hat{r}\ket{\Psi_{n}} - \bra{\Psi_{0}}\hat{s}^\dag\hat{r}\ket{\Psi_{0}} \\
&= F_n[\hat{s}^\dag\hat{r}].
\end{align*}
\noindent In the conditions of Proposition \ref{prop:recu_EOM_simple}, it is therefore sufficient to prove
$$\forall k \in \mathbb{N},\, F_n[\hat{O}_k] = F_n[\hat{s}^\dag\hat{r}]$$
for proving Proposition \ref{prop:recu_EOM_simple}. We will proceed by induction. Proving the base case ($k=0$) is immediate due to the definition of $\hat{O}_0$. Assuming that, for some natural integer $k$, we accept the 
\begin{equation}\label{eq:stepcaseEOM_simple}
F_n[\hat{O}_k] = F_n[\hat{s}^\dag\hat{r}]
\end{equation}
relation, we can show that
\begin{equation}\label{eq:stepcaseEOM_simplebis}
F_n[\hat{O}_{k+1}] = F_n[\hat{s}^\dag\hat{r}]
\end{equation}
can also be accepted — this is the step case. Indeed, we have
$$F_n[\hat{O}_{k+1}] = -\dfrac{1}{2}\braket{\Psi_0|[\hat{T}_{0\rightarrow n}^\dag,[[\hat{T}_{0\rightarrow n}^\dag,[\hat{O}_{k},\hat{T}_{0\rightarrow n}]],\hat{T}_{0\rightarrow n}]]|\Psi_0}$$
which, when developed, gives
$$F_n[\hat{O}_{k+1}] = -\dfrac{1}{2}\braket{\Psi_n|[\hat{T}_{0\rightarrow n}^\dag,[\hat{O}_{k},\hat{T}_{0\rightarrow n}]]|\Psi_n} - \left(-\dfrac{1}{2} \braket{\Psi_0|[\hat{T}_{0\rightarrow n}^\dag,[\hat{O}_{k},\hat{T}_{0\rightarrow n}]]|\Psi_0}\right).$$
Noticing that the sum of the two terms in the right-hand side of the last equality gives $F_n[\hat{O}_k]$, and invoking \eqref{eq:stepcaseEOM_simple}, we see that we are endowed for accepting \eqref{eq:stepcaseEOM_simplebis}. This concludes the proof.
\end{proof}
\noindent We now build $(\tilde{\boldsymbol{\gamma}}^\Delta_{k})_{k\in\mathbb{N}}^\wdag$, an infinite sequence of $L\times L$ matrices that are such that, for every $(r,s)$ couple of natural integers — both belonging to $\llbracket 1, L \rrbracket$ —, there is an infinite sequence of operators, $(\hatcal{O}_\ell^\wdag)^\wdag_{\ell\in\mathbb{N}}$, defined as
\begin{equation}\label{eq:relation_recu_RPA}
\hatcal{O}_0 := \hat{s}^\dag\hat{r} \;\, \mathrm{and} \;\, \forall \ell \in \mathbb{N},\, \hatcal{O}_{\ell+1} \coloneqq -\dfrac{1}{2} \left[ \hat{\mathcal{T}}^\dag, \left[\hatcal{O}_\ell,  \hat{\mathcal{T}}\right] \right]
\end{equation}
such that
\begin{equation*}
\forall k \in \mathbb{N},\, (\tilde{\boldsymbol{\gamma}}^\Delta_{k})_{r,s} \coloneqq \mathcal{F}[\hatcal{O}_k].
\end{equation*}
When dealing with the exact case, we saw with Proposition \ref{prop:recu_EOM_simple} that the choice for the argument of $F_n$ — as long as we use an element of the infinite sequence built in the proposition — is arbitrary. For building the \textit{reference} one-body reduced difference density matrix for the TDHF and TDDFRT methods we took \eqref{eq:EOM_1DDM} and replaced $\Psi_0$ by $\psi_\star$, and $\hat{T}_{0\rightarrow n}$ by $\hatcal{T}$. That is to say, according to \eqref{eq:EOMisFA}, we actually see that each \textit{reference} one-body reduced difference density matrix element value for the TDHF and TDDFRT methods is the $\mathcal{F}$ functional value with the first element of each $\hatcal{O}$ operator sequence built in \eqref{eq:relation_recu_RPA} as argument for $\mathcal{F}$. We wonder whether this choice is arbitrary, i.e., if replacing $\Psi_0$ by $\psi_\star$, and $\hat{T}_{0\rightarrow n}$ by $\hatcal{T}$ in other functionals than $F_n[\hat{O}_0]$ — which is equivalent to selecting other possible choices for the argument of $\mathcal{F}$ — would lead to matrices that one cannot reject based on physical criteria. This would question the uniqueness of the representation of TDHF and/or TDDFRT molecular electronic transitions.

\section{First approach — Directly addressing the problem}

\begin{proposition}\label{prop:gamma1nonsymmetric}
The $\tilde{\boldsymbol{\gamma}}^\Delta_{1}$ matrix is not symmetric.
\end{proposition}
\begin{proof}
Applying the derivation method introduced in Ref. \cite{etienne_comprehensive_2021}, we obtain
\begin{align}\label{eq:gamma_1}
\tilde{\boldsymbol{\gamma}}^\Delta_{1} = \dfrac{1}{2}
&\left(
-2\textbf{X} \textbf{X}^\top \textbf{X} \textbf{X}^\top
-2\textbf{Y} \textbf{Y}^\top \textbf{X} \textbf{X}^\top
-2\textbf{Y} \textbf{Y}^\top \textbf{Y} \textbf{Y}^\top
- \textbf{Y} \textbf{X}^\top \textbf{X} \textbf{Y}^\top
- \textbf{X} \textbf{Y}^\top \textbf{Y} \textbf{X}^\top\right)\nonumber\\
\oplus &\left(
 2\textbf{X}^\top \textbf{X} \textbf{X}^\top \textbf{X} 
+2\textbf{X}^\top \textbf{X} \textbf{Y}^\top \textbf{Y}
+2\textbf{Y}^\top \textbf{Y} \textbf{Y}^\top \textbf{Y}
+ \textbf{Y}^\top \textbf{X} \textbf{X}^\top \textbf{Y}
+ \textbf{X}^\top \textbf{Y} \textbf{Y}^\top \textbf{X}
\right)
\end{align}
\noindent which is not symmetric. Indeed, the $\textbf{Y} \textbf{Y}^\top \textbf{X} \textbf{X}^\top$ matrix, which appears in the North-West block, is not symmetric, and $\textbf{X} \textbf{X}^\top \textbf{Y} \textbf{Y}^\top$ is absent of the North-West block.
\end{proof}

\subsection{Fermi-vacuum diagrammatic evaluation of expectation values of nested commutators}

The difficulty we have for addressing our problem is that we are approaching the matrix representation of one operator — the one-body reduced difference density operator, belonging to $\mathrm{End}(\mathrm{span}_\mathbb{R}\mathcal{B})$ — in a basis of one-body states, namely $\mathcal{B}$, by evaluating the expectation value of another operator, element of $\mathrm{End}(\mathcal{C}_\mathbb{R})$, relatively to the $\psi_\star$ $N$--body state. Therefore, directly addressing the problem in first and second quantization became rapidly materially intractable for us, and we had to move toward a diagrammatic language framework introduced in \cite{morere2026dycklanguagefermionicsecondII} and extended below for the purpose of our derivations.

In what follows we call ``excitation operator'' (respectively, ``deexcitation operator'') any operator, written $\hat{E}_p^q$ (respectively, $\hat{D}^p_q$) and defined as $\hat{q}^\dag\hat{p}$ (respectively, $\hat{p}^\dag\hat{q}$), where $p$ belongs to $\llbracket 1, N \rrbracket$ and $q$ belongs to $\llbracket (N+1),L\rrbracket$.
  
In this paragraph we are interested in a particular type of \textit{conditioned} nested commutators — see Rule IV.1 and Rule IV.5 in Ref. \cite{morere2026dycklanguagefermionicsecondII} —, that we will name ``1--CNC''. 

\begin{definition}[1--CNC] A {\normalfont 1--CNC} is a conditioned nested commutator in which every slot is occupied by a (de)excitation operator except one slot in the innermost commutator, which is occupied by a creation-annihilation operator chain — e.g., ``$\hat{s}^\dag\hat{r}\,$'' — in which each operator points at a spinorbital with a non-definite occupation number relatively to the Fermi vacuum. \end{definition}

\subsubsection{A short reminder}

The diagrammatic representation of such a 1--CNC involves opening brackets for representing deexcitation operators, closing brackets for representing excitation operators, and a dash for representing the $\hat{s}^\dag\hat{r}\,$ chain. This representation can be used to ``develop'' the 1--CNC — see Rule IV.6 in Ref. \cite{morere2026dycklanguagefermionicsecondII}. The diagrammatic representation and development of one 1--CNC is given in Example \ref{ex:list}.

\begin{example}\label{ex:list}
Let $(i,j)$ be a couple of integers, both belonging to $\llbracket 1,N\rrbracket$. 
Let $(a,b)$ be a couple of integers, both belonging to $\llbracket N+1,L\rrbracket$. 
Let $\hat{s}^\dag$ and $\hat{r}$ be two second quantization operators, each pointing at a spin-orbital with non-definite occupation number relatively to the Fermi vacuum. Consider the 
\begin{equation*}\label{eq:commu_D-A-E}
\left[ \left[ \hat{D}_a^i, \hat{s}^\dag\hat{r} \right], \hat{E}_j^b \right]
\end{equation*}
nested-commutator. The diagrammatic development of this {\normalfont 1--CNC} is illustrated in {\normalfont Figure \ref{fig:dev_commu_D-A-E}}.

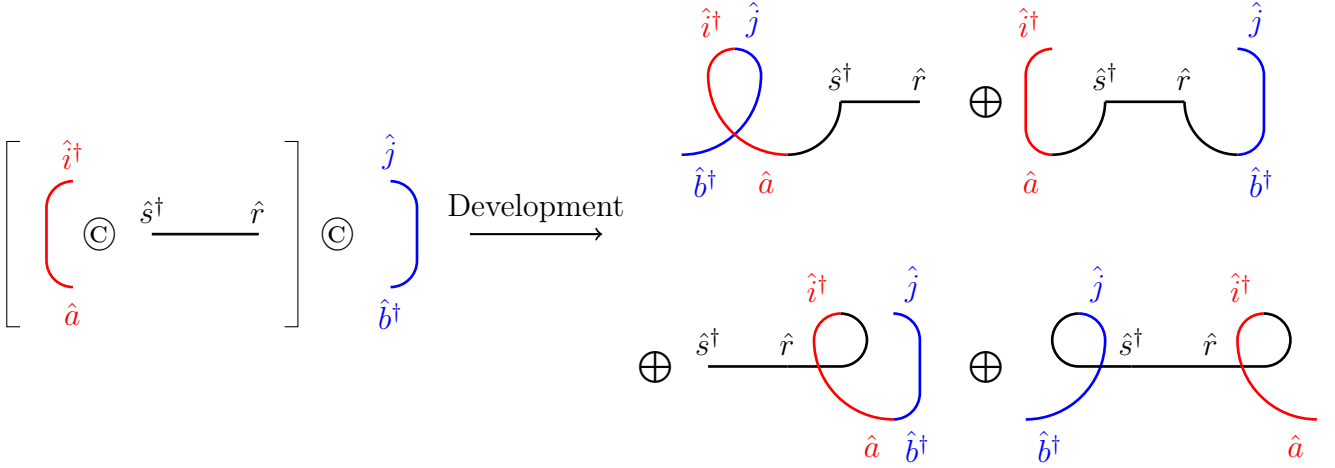
\begin{figure}[h!]
\begin{center}
\begin{tikzpicture}[scale=0.7]
\draw (-.5,1.75) -- ++(-.25,0) -- ++(0,-3.5) -- ++(.25,0);
\draw (1,0) node {$\copyright$};
\tcr{\pdeex{0}{\hat{i}^\dag}{\hat{a}^\wdag}}\padaga{3}{}{\hat{s}^\dag}{}{\hat{r}^\wdag} 
\draw (4.5,1.75) -- ++(.25,0) -- ++(0,-3.5) -- ++(-.25,0);
\tikzset{shift={(4,0)}}
\draw (1.5,0) node {$\copyright$};
\tikzset{shift={(2,0)}}
\tcb{\pexci{1}{\hat{j}^\wdag}{\hat{b}^\dag}}

\draw[->,thick] (2,0) --node[above] {Development} ++(2.5,0);
\tikzset{shift={(6,2.5)}}

\tcb{\exci{\hat{j}}{\hat{b}^\dag}}\tikzset{shift={(1,0)}}
\tcr{\deex{\hat{i}^\dag}{\hat{a}^\wdag}}\tikzset{shift={(1,0)}}
\adaga{exci}{\hat{s}^\dag}{none}{\hat{r}^\wdag}
\tikzset{shift={(2.5,0)}}
\draw (-1.25,0) node {$\bigoplus$};
\tcr{\deexend{\hat{i}^\dag}{\hat{a}^\wdag}}
\adaga{exci}{\hat{s}^\dag}{deex}{\hat{r}^\wdag}
\tcb{\exciend{\hat{j}^\wdag}{\hat{b}^\dag}}
\tikzset{shift={(-10,-5)}}
\draw (-1,0) node {$\bigoplus$};
\adaga{none}{\hat{s}^\dag}{exci}{\hat{r}^\wdag}
\tcr{\deex{\hat{i}^\dag}{\hat{a}^\wdag}}\tikzset{shift={(1,0)}}
\tcb{\exciend{\hat{j}^\wdag}{\hat{b}^\dag}}
\tikzset{shift={(2.5,0)}}
\draw (-.75,0) node {$\bigoplus$};
\tcb{\exci{\hat{j}^\wdag}{\hat{b}^\dag}}\tikzset{shift={(1,0)}}
\adaga{deex}{\hat{s}^\dag}{exci}{\hat{r}^\wdag}
\tcr{\deex{\hat{i}^\dag}{\hat{a}^\wdag}}\tikzset{shift={(1,0)}}
\end{tikzpicture}
\end{center}
\caption{Diagrammatic development of the $[ [\hat{D}_a^i, \hat{s}^\dag\hat{r} ], \hat{E}_j^b ]$ 1--CNC. The diagrammatic representation of the $\hat{D}_a^i$ deexcitation operator is highlighted in red and the diagrammatic representation of the $\hat{E}_j^b$ excitation operator is highlighted in blue.}
\label{fig:dev_commu_D-A-E}
\end{figure}
\end{example}

\noindent The diagrammatic simplification of a 1--CNC consists in successively connecting the opening and the closing brackets to the left or to the right of the dash — this is extensively introduced and illustrated in paragraph IV.3 of Ref. \cite{morere2026dycklanguagefermionicsecondII}. Each sub-diagram obtained in the diagrammatic development of a 1--CNC corresponds exactly to one operator in the diagrammatic simplification of that 1--CNC. For instance, the diagrammatic development of $[ [\hat{D}_a^i, \hat{s}^\dag\hat{r} ], \hat{E}_j^b ]$ reported in Figure \ref{fig:dev_commu_D-A-E} involves four sub-diagrams — see Rule IV.3 of Ref. \cite{morere2026dycklanguagefermionicsecondII}. The diagrammatic simplification of the $[ [\hat{D}_a^i, \hat{s}^\dag\hat{r} ], \hat{E}_j^b ]$ 1--CNC then reads
$$[ [\hat{D}_a^i, \hat{s}^\dag\hat{r} ], \hat{E}_j^b ] = -\delta_{i,j}\delta_{a,s}\hat{b}^\dag\hat{r} + \delta_{a,s}\delta_{r,b}\hat{i}^\dag\hat{j} - \delta_{r,i}\delta_{a,b}\hat{s}^\dag\hat{j} + \delta_{j,s}\delta_{r,i}\hat{b}^\dag\hat{a},$$
where $-\delta_{i,j}\delta_{a,s}\hat{b}^\dag\hat{r}$ corresponds to the top-left sub-diagram at the right of Figure \ref{fig:dev_commu_D-A-E},
$\delta_{a,s}\delta_{r,b}\hat{i}^\dag\hat{j}$ corresponds to the top-right sub-diagram, $- \delta_{r,i}\delta_{a,b}\hat{s}^\dag\hat{j}$ corresponds to the bottom-left sub-diagram, and $\delta_{j,s}\delta_{r,i}\hat{b}^\dag\hat{a}$ corresponds to the bottom-right sub-diagram.

\subsubsection{Necessary conditions for expectation values of 1--CNC's to be non-zero}

Each sub-diagram in the diagrammatic simplification of a 1--CNC can be uniquely defined by two sequences: an \textit{operator} sequence and a \textit{connection} sequence.

\begin{definition}[Operator sequence of a 1--CNC]\label{def:seq_operator}
The {\normalfont{operator sequence of a 1--CNC}} is the sequence of (de)excitation operators encountered from the innermost to the outermost commutator.
\end{definition}

\noindent Knowing that we are dealing with 1--CNC's, any sub-diagram that can be generated during a 1--CNC development necessarily involves a curved line with a left and a right extremity.

\begin{definition}[Connection sequence of a sub-diagram corresponding to a 1--CNC]\label{def:seq_connection}
Let $\hatcal{A}$ be a {\normalfont 1--CNC}. Let $O({\hatcal{A}})$ be the corresponding operator sequence. The {\normalfont connection sequence of a sub-diagram} obtained when diagrammatically developing $\hatcal{A}$ is a sequence of one-letter elements, with letters taken in the $\{\text{L},\text{R}\}$ set. The $n^\text{th}$ element of a {{connection sequence}} is an ``R'' (respectively, ``L'') if the diagrammatic representation of the $n^\text{th}$ element of $O({\hatcal{A}})$ is placed at the right (respectively, left) of the dash in the considered sub-diagram.
\end{definition}

\noindent In Example \ref{ex:list}, the operator sequence corresponding to the $[ [\hat{D}_a^i, \hat{s}^\dag\hat{r} ], \hat{E}_j^b ]$ 1--CNC is $(\hat{D}_a^i, \hat{E}_j^b)$. The diagrammatic development of that 1--CNC involves four sub-diagrams, placed at the right of Figure \ref{fig:dev_commu_D-A-E}.

In the $(\hat{D}_a^i, \hat{E}_j^b)$ operator sequence, the $\hat{D}_a^i$ operator comes \textit{first}. Following the definition of a connection sequence — see Definition \ref{def:seq_connection} — if, in a given sub-diagram, the diagrammatic representation of the $\hat{D}_a^i$ operator is at the left (respectively, right) of the dash, the \textit{first} element of the connection sequence of that sub-diagram is an ``\textit{L}'' (respectively, ``\textit{R}'').

In the $(\hat{D}_a^i, \hat{E}_j^b)$ operator sequence, the $\hat{E}_j^b$ operator comes at the \textit{second} place. Following the definition of a connection sequence — see Definition \ref{def:seq_connection} — if, in a given sub-diagram, the diagrammatic representation of the $\hat{E}_j^b$ operator is at the left (respectively, right) of the dash, the \textit{second} element of the connection sequence of that sub-diagram is an ``\textit{L}'' (respectively, ``\textit{R}'').

Applying these rules to the diagrammatic development of the $[ [\hat{D}_a^i, \hat{s}^\dag\hat{r} ], \hat{E}_j^b ]$ 1--CNC reported in Figure \ref{fig:dev_commu_D-A-E} we find that we can attribute the following connection sequences to the four sub-diagrams in the right part of the figure: (\textit{L},\textit{L}) for the top-left sub-diagram, (\textit{L},\textit{R}) for the top-right sub-diagram, (\textit{R},\textit{R}) for the bottom-left sub-diagram, (\textit{R},\textit{L}) for the bottom-right sub-diagram.

\begin{proposition}\label{prop:nullity_criterion_sub_diagram}
If both extremities of a sub-diagram obtained from the diagrammatic development of a {\normalfont 1--CNC} are not in the upper half-plane, then the Fermi-vacuum expectation value of the corresponding operator is equal to zero.
\end{proposition}
\begin{proof}
The operator corresponding to a sub-diagram obtained from the diagrammatic development of a 1--CNC involves a creation-annihilation chain. These two operators correspond to the extremities in the sub-diagram. The Fermi-vacuum expectation value of a creation-annihilation chain may not be equal to zero if both operators point to an occupied spin-orbital. Therefore, for the expectation value to be different from zero, both extremities have to be in the upper half-plane --- which corresponds to the occupied spin-orbital space.
\end{proof}

\begin{proposition}\label{prop:change_half_plan}
A sub-diagram obtained from the diagrammatic development of a {\normalfont 1--CNC}, is composed of two (possibly empty) alternating sequences of opening and closing brackets, one on the left and one the right side of the dash.
\end{proposition}
\begin{proof}
During the construction of the sub-diagram, two extremities can be connected if they point to different directions — see, in Ref. \cite{morere2026dycklanguagefermionicsecondII}, paragraph II.2.2, together with Rule IV.2 and Proposition II.1. Since the extremities of opening brackets point to the right, whereas those of closing brackets point to the left, successive connected brackets must alternate between opening and closing brackets. The extremities of the dash can bend in any direction and can be connected to both opening and closing brackets.
\end{proof}

\begin{proposition}\label{prop:odd_bracket}
Consider a sub-diagram, obtained from the diagrammatic development of a {\normalfont 1--CNC}, with a connection sequence containing an odd number of right (respectively, left) connections. The Fermi-vacuum expectation value of the operator corresponding to that sub-diagram may be different from zero if the first bracket connected to the right (respectively, left) extremity of the dash is a closing (respectively, opening) bracket. 
\end{proposition}
\begin{proof}
According to {\normalfont Proposition \ref{prop:nullity_criterion_sub_diagram}}, for the Fermi-vacuum expectation value of the operator corresponding to the sub-diagram to be susceptible of being different from zero, the sub-diagram's right extremity has to be in the upper half-plane. Therefore, the last bracket on the right of the sub-diagram has to be a closing one, as it is the only possibility that leads to a right extremity in the upper half-plane.

According to {\normalfont Proposition \ref{prop:change_half_plan}}, for an odd number of brackets at the right of the dash, the fact that the rightmost bracket is a closing one implies that the first bracket connected to the dash is a closing bracket too.

The argument is also valid for the left extremity, but in that case the ``closing'' word has to be changed into ``opening'' in what precedes in the proof.
\end{proof}
\begin{remark}
When, in {\normalfont Proposition \ref{prop:odd_bracket}}, we say that ``the Fermi-vacuum expectation value of the operator {\normalfont[...]} may different from zero'', we mean that the condition reported in the proposition is a necessary but not sufficient condition — i.e., it is possible that when assigning numbers to the operator indices, the chosen combination of indices leads to a zero expectation value.
\end{remark}
\begin{proposition}\label{prop:even_bracket}
Consider a sub-diagram, obtained from the diagrammatic development of a {\normalfont 1--CNC}, with a connection sequence containing an even number of right (respectively, left) connections. The Fermi-vacuum expectation value of the operator corresponding to that sub-diagram may be different from zero if the first bracket connected to the right (respectively, left) extremity of the dash is a opening (respectively, closing) bracket. 
\end{proposition}
\begin{proof}
The proof follows the same reasoning as for {\normalfont Proposition \ref{prop:odd_bracket}}. The difference is that an even number of connections on a side of the dash implies that, for that side, the first and last brackets are of different nature (closing vs. opening).
\end{proof}

\begin{corollary}\label{remark:occupancy_rs}
Consider a sub-diagram corresponding to an operator whose Fermi-vacuum expectation value is susceptible of being different from zero. If there is an odd (respectively, even) number of ``R'' entries in its connection sequence, then the right extremity of the dash bends in the lower (respectively, upper) half-plane. Similarly, if there is an odd (respectively, even) number of ``L'' entries in its connection sequence, then the left extremity of the dash bends in the lower (respectively, upper) half-plane.
\end{corollary}
\begin{proof}
This is a direct consequence of propositions {\normalfont \ref{prop:odd_bracket} and \ref{prop:even_bracket}}.
\end{proof}

\begin{corollary}\label{remark:number_crossing}
Consider a sub-diagram corresponding to an operator whose Fermi-vacuum expectation value is susceptible of being different from zero. If there are $2m$ ``R'' (respectively, ``L'') entries in its connection sequence — with $m$ being a natural integer —, then there are $m$ crossings on the right (respectively, left) of the dash. If there are $(2m+1)$``R'' (respectively, ``L'') entries in its connection sequence — with $m$ being a natural integer —, then there are $m$ crossings on the right (respectively, left) of the dash. 
\end{corollary}

\begin{proof}
This is a direct consequence of propositions {\normalfont \ref{prop:change_half_plan}, \ref{prop:odd_bracket}, and \ref{prop:even_bracket}}.
\end{proof}

\begin{example}
In {\normalfont Example \ref{ex:list}}, the top-left and bottom-right sub-diagrams in {\normalfont Figure \ref{fig:dev_commu_D-A-E}} do not have both extremities in the upper half-plane. Hence, according to {\normalfont Proposition \ref{prop:nullity_criterion_sub_diagram}} the Fermi-vacuum expectation value of their corresponding operator is equal to zero.  

Moreover we can notice that the connection sequence $(L,R)$ {\normalfont(}respectively, $(R,R)${\normalfont)}, corresponding to the top-right (respectively, bottom-left) diagram contains an odd (respectively, even) number of ``R'' and ``L''.  Then, according to {\normalfont Corollary \ref{remark:occupancy_rs}}, both extremities of the dash must bend to the lower (respectively, upper) half-plane for the Fermi-vacuum expectation value of the corresponding operator to be susceptible of being non-zero. 

Accordingly, we can provide the value of the Fermi-vacuum expectation value of the $[ [\hat{D}_a^i, \hat{s}^\dag\hat{r} ], \hat{E}_j^b ]$ {\normalfont 1--CNC} as
\begin{equation*}
\left\langle \Psi_{0} \left| \left[ \left[ \hat{D}_a^i, \hat{s}^\dag\hat{r} \right], \hat{E}_j^b \right] \right| \Psi_{0} \right\rangle = \delta_{r,b}\delta_{i,j}\delta_{a, s}(1-n_r)(1-n_s) - \delta_{r, i}\delta_{a, b}\delta_{s, j}n_rn_s.
\end{equation*}
The diagrammatic evaluation of the Fermi-vacuum expectation value of $[ [\hat{D}_a^i, \hat{s}^\dag\hat{r} ], \hat{E}_j^b ]$ is provided in {\normalfont Figure \ref{fig:ev_Psi0_1CNC}}.
\begin{figure}[h!]
\begin{center}
\begin{tikzpicture}[scale=0.7]
\draw[->,thick] (0,0) --node[above] {No allowed} node[below] {connection} ++(2.5,0);
\tikzset{shift={(-7,0)}}
\exci{\hat{j}}{\hat{b}^\dag}\deex{\hat{i}^\dag}{\hat{a}}\adaga{exci}{\hat{s}^\dag}{none}{\hat{r}}

\tikzset{shift={(2.5,-5)}}

\draw[->,thick] (0,0) --node[above] {Connection of} node[below] {$\hat{i}^\dag$ and $\hat{j}$} ++(2.5,0);
\tikzset{shift={(-7,0)}}
\deexend{\hat{i}^\dag}{\hat{a}}\adaga{exci}{\hat{s}^\dag}{deex}{\hat{r}}\exciend{\hat{j}}{\hat{b}^\dag}

\tikzset{shift={(8,0)}}

\deexend{\hat{i}^\dag}{\hat{a}}\adaga{exci}{\hat{s}^\dag}{deex}{\hat{r}}\exciend{\hat{j}}{\hat{b}^\dag}
\draw[line width=1pt] (0,1) -- ++(-3.5,0);

\tikzset{shift={(-8,-5)}}

\draw[->,thick] (0,0) --node[above] {Connection of} node[below] {$\hat{s}^\dag$ and $\hat{j}$} ++(2.5,0);
\tikzset{shift={(-7,0)}}
\adaga{none}{\hat{s}^\dag}{exci}{\hat{r}}\deex{\hat{i}^\dag}{\hat{a}}\exciend{\hat{j}}{\hat{b}^\dag}

\tikzset{shift={(8,0)}}

\draw[line width=1pt] (-.5,1) -- (-.5,.5) arc(180:270:.5);
\draw[line width=1pt] (0,0)node[above]{$\hat{s}^\dag$} -- (1.5,0)node[above]{$\hat{r}^\wdag$};
\tikzset{shift={(1.5,0)}}
\draw[line width=1pt] (0,0) -- (1,0) arc(-90:90:.5);
\draw[line width=1pt] (1,1) node[above left]{$\hat{i}^\dag$} arc (90:180:.5) 
          --(.5,.5) arc (180:270:1.5)node[below left]{$\hat{a}^\wdag$};
\draw[line width=1pt] 
             (2,-1)  node[below right]{$\hat{b}^\dag$} arc (-90:0:.5) 
          -- (2.5,1) node[above right]{$\hat{j}^\wdag$} arc (0:180:2.25 and 1);

\tikzset{shift={(-6,-5)}}

\draw[->,thick] (0,0) --node[above] {No allowed} node[below] {connection} ++(2.5,0);
\tikzset{shift={(-7.5,0)}}
\exci{\hat{j}}{\hat{b}^\dag}\adaga{deex}{\hat{s}^\dag}{exci}{\hat{r}}\deex{\hat{i}^\dag}{\hat{a}}
\end{tikzpicture}
\end{center}
\caption{Diagrammatic evaluation of the Fermi-vacuum expectation value of $[ [\hat{D}_a^i, \hat{s}^\dag\hat{r} ], \hat{E}_j^b ]$. Among the four sub-diagrams obtained in Figure \ref{fig:dev_commu_D-A-E}, only two can be well-fully connected.}
\label{fig:ev_Psi0_1CNC}
\end{figure}
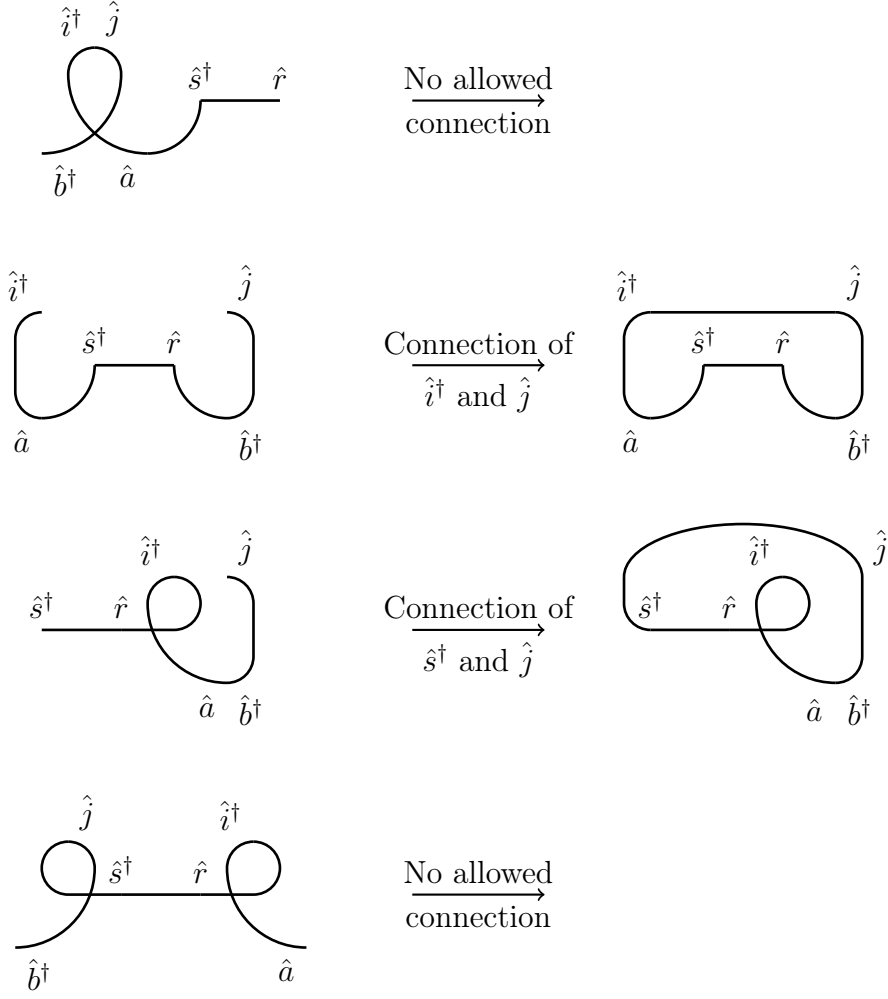
\end{example}
\noindent In Proposition \ref{prop:non_nullity} below, we first consider the \textit{nature} (closing vs. opening) of the brackets in (sub-)diagrams, and the \textit{nature} (excitation vs. deexcitation) of the operators at stake. 

\begin{definition}[Undressed operator sequence of a 1--CNC]\label{def:seq_undressedoperator}
The {\normalfont{undressed operator sequence of a 1--CNC}} is the sequence of non-indexed (de)excitation operators encountered from the innermost to the outermost commutator.
\end{definition}

\begin{proposition}\label{prop:non_nullity}
For a given connection sequence, there exists a unique undressed operator sequence which is such that, when dressed, generates together with the connection sequence a sub-diagram corresponding to an operator which has an expectation value relatively to the Fermi vacuum that is susceptible of being non-zero.
\end{proposition}
\begin{proof}
The nature (opening or closing) of the leftmost and rightmost brackets in a sub-diagram corresponding to an operator whose Fermi-vacuum expectation value is susceptible of being different from zero is imposed by {\normalfont Proposition \ref{prop:nullity_criterion_sub_diagram}} --- the leftmost bracket is an opening one, and the rightmost bracket is a closing one. The nature (opening or closing) of all the other brackets is fully determined from {\normalfont Proposition \ref{prop:change_half_plan}}. The connection sequence imposes the order of appearance of the brackets in the sub-diagram. In other words, the connection sequence uniquely determines the nature and order of every bracket that have to be involved in the sub-diagram to lead to an operator whose expectation value relatively to the Fermi vacuum may be non-zero.
Since every opening (respectively, closing) bracket corresponds to a deexcitation (respectively, an excitation) operator, the connection sequence determines a unique undressed operator sequence which, when dressed, generates together with the connection sequence a sub-diagram corresponding to an operator which has an expectation value relatively to the Fermi vacuum that is susceptible of being non-zero.
\end{proof}

\subsection{Application to the TDHF/TDDFRT case}

We set
$$\hat{\mathcal{T}}_x \coloneqq \sum_{i=1}^{N}\sum_{a=N+1}^{L}
 \textbf{x}_{ia} \hat{a}^\dag\hat{i} , $$
and
$$ \hat{\mathcal{T}}_y \coloneqq \sum_{i=1}^{N}\sum_{a=N+1}^{L}
 \textbf{y}_{ia} \hat{i}^\dag\hat{a}.$$
The operator present in the $m^\text{th}$ commutator of a 1--CNC — commutators are counted from the innermost to the outermost one — reads
\begin{equation}\label{eq:nomenclatureT}
\mathcal{\hat{T}}  = 
\sum_{i_m=1}^{N}\sum_{a_m=N+1}^{L}
 (\textbf{x}_{i_ma_m} \hat{a}_m^\dag\hat{i}_m^\wdag
-\textbf{y}_{i_ma_m} \hat{i}_m^\dag\hat{a}_m^\wdag)
=
\mathcal{\hat{T}}_x- \mathcal{\hat{T}}_y
\end{equation}
\begin{equation}\label{eq:nomenclatureTdag}
\mathcal{\hat{T}}^\dag = 
\sum_{j_m=1}^{N}\sum_{b_m=N+1}^{L} 
 (\textbf{x}_{j_mb_m} \hat{j}^\dag_m\hat{b}_m^\wdag
-\textbf{y}_{j_mb_m} \hat{b}_m^\dag\hat{j}_m^\wdag)
=
\mathcal{\hat{T}}_x^\dag- \mathcal{\hat{T}}_y^\dag
\end{equation}
In what follows, and consistently with what we have already done in \eqref{eq:nomenclatureT} and \eqref{eq:nomenclatureTdag}, for the operators present in the $m^\text{th}$ commutator, we use the following nomenclature: ``$i$'' and ``$a$'' are used if the operator is issued from a $\hat{\mathcal{T}}$ operator; ``$j$'' and ``$b$'' are used if the operator is issued from a $\hat{\mathcal{T}}^\dag$ operator; ``$l$'' and ``$c$'' are used if the operator is issued from a non-explicited ($\hat{\mathcal{T}}$ or $\hat{\mathcal{T}}^\dag$) operator. The explicit cases are summarized in Figure \ref{fig:nomenclature_bracket}.

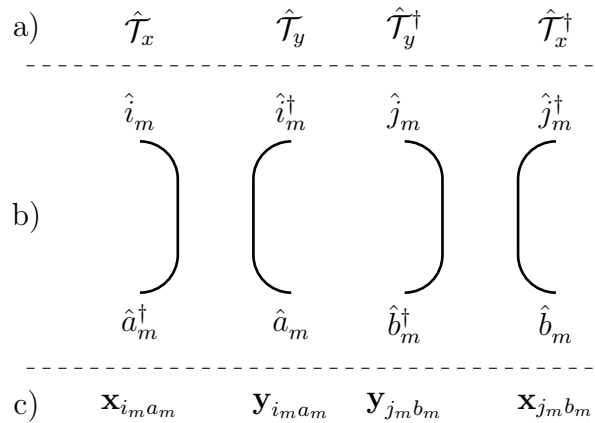
\begin{figure}[h!]
\begin{center}
\begin{tikzpicture}[scale=1]

\node (Op)  at (-2,0)  {a)};
\node (Tx)  at (-.5,0)  {$\hat{\mathcal{T}}_x$};
\node (Ty)  at (1.5,0)  {$\hat{\mathcal{T}}_y$};
\node (Tyt) at (3,0)    {$\hat{\mathcal{T}}_y^\dag$};
\node (Txt) at (5,0)    {$\hat{\mathcal{T}}_x^\dag$};

\draw[dashed] (-2,-.5) -- (5.5,-.5);
\tikzset{shift={(0,-2.5)}}
\node (diag)  at (-2,0)  {b)};
\pexci{0}{\hat{i}_m^\wdag}{\hat{a}_m^\dag}
\pdeex{1}{\hat{i}_m^\dag}{\hat{a}_m^\wdag}
\pexci{3.5}{\hat{j}_m^\wdag}{\hat{b}_m^\dag}
\pdeex{4.5}{\hat{j}_m^\dag}{\hat{b}_m^\wdag}

\draw[dashed] (-2,-2) -- (5.5,-2);
\tikzset{shift={(0,-2.5)}}
\node (mat)  at (-2,0)  {c)};
\node (x)  at (-.5,0)  {$\textbf{x}_{i_ma_m}$};
\node (y)  at (1.5,0)  {$\textbf{y}_{i_ma_m}$};
\node (yt) at (3,0)    {$\textbf{y}_{j_mb_m}$};
\node (xt) at (5,0)    {$\textbf{x}_{j_mb_m}$};

\end{tikzpicture}
\end{center}
\caption{a) Operator in the $m^\text{th}$ commutator, b) contribution to the diagrammatic representation, c) coefficient contribution to the matrix representation.}
\label{fig:nomenclature_bracket}
\end{figure}

\begin{lemma}\label{lemma:conditioning}
For every $k$ in $\mathbb{N}$, any $\mathcal{F}[\hat{\mathcal{O}}_k]$ functional value can be decomposed into a sum of non-zero real multiple of {\normalfont 1--CNC}'s Fermi-vacuum expectation values.
\end{lemma}
\begin{proof}
\noindent For every $k$ in $\mathbb{N}$, any commutator in $\hat{\mathcal{O}}_k$ involving a $\mathcal{\hat{T}}$ operator has this operator placed in the right slot of the commutator — see \eqref{eq:relation_recu_RPA} — and therefore reads
\begin{equation}\label{eq:comCTxTy}
\left[\hat{C},\hat{\mathcal{T}}_x-\hat{\mathcal{T}}_y\right] = \left[\hat{C},\hat{\mathcal{T}}_x\right] + \left[\hat{\mathcal{T}}_y,\hat{C}\right],
\end{equation}
where $\hat{C}$ is another operator, possibly a nested commutator. We know that $\hat{\mathcal{T}}_x$ contains excitation operators and $\hat{\mathcal{T}}_y$ contains deexcitation operators. For that reason, we see that the two commutators in the right-hand side of \eqref{eq:comCTxTy} are readily \textit{conditioned} — see Rule IV.1 in \cite{morere2026dycklanguagefermionicsecondII}.

For every $k$ in $\mathbb{N}$, any commutator in $\hat{\mathcal{O}}_k$ involving a $\mathcal{\hat{T}}^\dag$ operator has this operator placed in the left slot of the commutator — see \eqref{eq:relation_recu_RPA} — and therefore reads 
\begin{equation}\label{eq:comCTxTydag}
\left[\hat{\mathcal{T}}_x^\dag-\hat{\mathcal{T}}_y^\dag,\hat{C}\right] = \left[\hat{\mathcal{T}}_x^\dag,\hat{C}\right] + \left[\hat{C},\hat{\mathcal{T}}_y^\dag\right]
\end{equation}
where $\hat{C}$ is another operator, possibly a nested commutator. We know that $\hat{\mathcal{T}}_x^\dag$ contains deexcitation operators and $\hat{\mathcal{T}}_y^\dag$ contains excitation operators. For that reason, we see that the two commutators in the right-hand side of \eqref{eq:comCTxTydag} are readily \textit{conditioned} — see Rule IV.1 in \cite{morere2026dycklanguagefermionicsecondII}. This completes the proof.
\end{proof}

\begin{remark}
In what precedes, we used the $(r,s)$ couple of integers for pointing at couples of spin-orbitals in $\mathcal{B}$ whose Fermi-vacuum occupation number was non-definite. In the context of propositions {\normalfont \ref{prop:existYYYYXX}} and {\normalfont \ref{prop:absenceXXYY}}, the same couple of integers will be used for pointing at non-definite couples of spin-orbitals taken among those in $\mathcal{B}$ who have a Fermi-vacuum occupation number equal to one.
\end{remark}
\begin{proposition}\label{prop:existYYYYXX}
For every natural integer $k$ strictly greater than one, the North-West — i.e., the occupied--occupied — block of $\tilde{\boldsymbol{\gamma}}^\Delta_{k}$ contains a strictly negative multiple of
\begin{equation}\label{eq:YYYYXX}
{\normalfont\underbrace{\textbf{Y}\textbf{Y}^\top \cdots \textbf{Y}\textbf{Y}^\top}_{\textbf{Y}\textbf{Y}^\top \,\text{appears}\, k\,\text{times}} \textbf{X} \textbf{X}^\top ,}
\end{equation}
which is not symmetric.
\end{proposition}

\begin{proof}
Let $k$ be a natural integer strictly greater than one. The $\mathcal{F}[\hat{\mathcal{O}}_k]$ number is the Fermi-vacuum expectation value of a nested commutator written using ``$\hat{\mathcal{T}}$'' and ``$\hat{\mathcal{T}}^\dag$''. That nested commutator can be developed as a sum of nested commutators involving solely $\hat{s}^\dag \hat{r}$ and $\hat{\mathcal{T}}_x^\wdag$, $\hat{\mathcal{T}}_x^\dag$, $\hat{\mathcal{T}}_y^\wdag$, and $\hat{\mathcal{T}}_y^\dag$. One of these nested commutators is
\begin{equation}\label{1CNC_YYYYXX}
\underbrace{\left[\left[\hat{\mathcal{T}}_y,\cdots\left[\left[\hat{\mathcal{T}}_y,\right.\right.\right.\right.}_{\text{``}\left[\left[\hat{\mathcal{T}}_y,\right.\right.\text{''} \; \mathrm{appears} \; k \; \mathrm{times}}\left[\hat{\mathcal{T}}_x^\dag,\left[\hat{s}^\dag\hat{r}^\wdag,\hat{\mathcal{T}}_x\right]\right]\underbrace{\left.\left.\left],\hat{\mathcal{T}}_y^\dag\right]\cdots\,\right],\hat{\mathcal{T}}_y^\dag\right]}_{\text{``}\left],\hat{\mathcal{T}}_y^\dag\right]\text{''} \; \mathrm{appears} \; k \; \mathrm{times}}.
\end{equation}
Due to the bilinearity of commutators, we can factorize the \textbf{x} and \textbf{y} coefficients coming from the $\hat{\mathcal{T}}_x^\wdag$, $\hat{\mathcal{T}}_x^\dag$, $\hat{\mathcal{T}}_y^\wdag$, and $\hat{\mathcal{T}}_y^\dag$ operators.

The operator sequence corresponding to the remaining 1--CNC is
\begin{equation}\label{eq:OseqYYXX}
O = (\hat{E}_{i_1}^{a_1}, \hat{D}_{b_2}^{j_2}, \hat{D}_{a_3}^{i_3}, \hat{E}_{j_4}^{b_4}, \ldots, \hat{D}_{a_{2k+1}}^{i_{2k+1}}, \hat{E}_{j_{2k+2}}^{b_{2k+2}}),
\end{equation}
where, after the second element, we have a repeating deexcitation-excitation motif until the end of the sequence.
Consider the 
\begin{equation}\label{eq:CseqYYXX}
C = (\textit{L},\textit{L},\textit{R},\textit{R},\ldots,\textit{R},\textit{R}) 
\end{equation}
connection sequence. The sub-diagram generated by $C$ and $O$ is represented in Figure \ref{fig:subdiag_YYYYXX} and corresponds to the
\begin{equation*}
{(-1)^{\ell}}\delta_{b_2, a_1}\delta_{i_1, s}\delta_{r, i_3}\delta_{a_3, b_4}\cdots\delta_{a_{2k+1}, b_{2k+2}}\hat{j}_2^\dag\hat{j}_{2k+2}^\wdag
\end{equation*}
operator, where $\ell$ is the number of crossings in the sub-diagram — \textit{vide infra}.
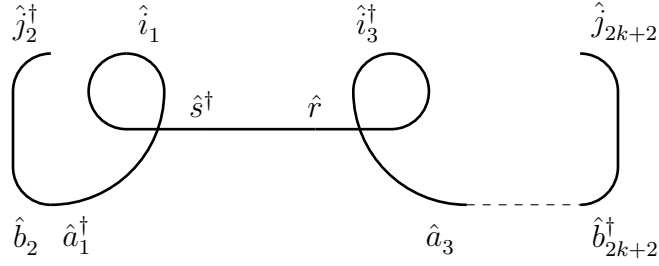
\begin{figure}[h!]
\begin{center}
\begin{tikzpicture}
\deexend{\hat{j}_{2}^\dag}{\hat{b}_{2}^\wdag}
\exci{\hat{i}_{1}^\wdag}{\hat{a}_{1}^\dag}
\adaga{deex}{\hat{s}^\dag}{exci}{\hat{r}^\wdag}\deex{\hat{i}_3^\dag}{\hat{a}_3^\wdag}
\draw[dashed] (0,-1) -- (1.5,-1);
\tikzset{shift={(1.5,0)}}
\exciend{\hat{j}_{2k+2}^\wdag}{\hat{b}_{2k+2}^\dag}
\end{tikzpicture}
\end{center}
\caption{Sub-diagram generated by the $O$ operator sequence from \eqref{eq:OseqYYXX}, and the $C$ connection sequence from \eqref{eq:CseqYYXX}.}
\label{fig:subdiag_YYYYXX}
\end{figure}
As a consequence of what precedes, in the expansion and simplification of \eqref{1CNC_YYYYXX}, the
\begin{align}\label{eq:decompositionYYYYXX}
\hspace*{-.5cm}(-1)^{\ell}
\sum_{\substack{(i_u)_{u\in \mathcal{I}}\\(j_u)_{u\in \mathcal{P}}}}^{\llbracket 1,N\rrbracket}
\sum_{\substack{(a_u)_{u\in \mathcal{I}}\\(b_u)_{u\in \mathcal{P}}}}^{\llbracket (N+1),L\rrbracket}
&((\textbf{X})_{i_1,a_1-N}^\wdag
(\textbf{X}^\top)_{j_2,b_2-N}
(\textbf{Y})_{i_3,a_3-N}^\wdag
(\textbf{Y}^\top)_{j_4,b_4-N}
\cdots
(\textbf{Y})_{i_{2k+1},a_{2k+1}-N}^\wdag
(\textbf{Y}^\top)_{b_{2k+2}-N,j_{2k+2}}\nonumber\\
&\times\delta_{b_2, a_1}\delta_{i_1, s}\delta_{r, i_3}\delta_{a_3, b_4}\cdots\delta_{a_{2k+1}, b_{2k+2}}\hat{j}_2^\dag\hat{j}_{2k+2}^\wdag)
\end{align}
term will necessarily appear, where $\mathcal{P} \coloneqq (2z+2)_{z\in \llbracket 0, k\rrbracket}$ and $\mathcal{I} \coloneqq (2z+1)_{z\in \llbracket 0, k\rrbracket}$, and
$$\sum_{\substack{(i_u)_{u\in \mathcal{I}}\\(j_u)_{u\in \mathcal{P}}}}^{\llbracket 1,N\rrbracket} = \sum _{i_1 = 1}^N \sum_{i_3 = 1}^N \cdots \sum_{i_{2k+1}=1}^N \sum _{j_2 = 1}^N \sum_{j_4 = 1}^N \cdots \sum_{j_{2k+2}=1}^N$$
and
$$\sum_{\substack{(a_u)_{u\in \mathcal{I}}\\(b_u)_{u\in \mathcal{P}}}}^{\llbracket N+1,L\rrbracket} = \sum _{a_1 = N+1}^L \sum_{a_3 = N+1}^L \cdots \sum_{a_{2k+1}=N+1}^L \sum _{b_2 = N+1}^L \sum_{b_4 = N+1}^L \cdots \sum_{b_{2k+2}=N+1}^L$$
The Fermi-vacuum expectation value of \eqref{eq:decompositionYYYYXX} is equal to
\begin{equation*}
\hspace*{-.5cm}(-1)^{\ell}
\sum_{(j_u)_{u\in \mathcal{P}^*}}^{\llbracket 1,N\rrbracket}
\sum_{(b_u)_{u\in \mathcal{P}}}^{\llbracket (N+1),L\rrbracket}
(\textbf{Y})_{r,b_4-N}^\wdag(\textbf{Y}^\top)_{b_4-N,j_4}
\cdots
(\textbf{Y})_{j_{2k},b_{2k+2}-N}^\wdag(\textbf{Y}^\top)_{b_{2k+2}-N,j_{2k+2}}
(\textbf{X})_{j_{2k+2},b_2-N}^\wdag(\textbf{X}^\top)_{b_2-N,s} ,
\end{equation*}
i.e.,
$$(-1)^{\ell}(\textbf{Y}\textbf{Y}^\top \cdots \textbf{Y}\textbf{Y}^\top\textbf{X}\textbf{X}^\top)_{r,s}.$$
In the last step of this derivation, we have used $\mathcal{P}^* \coloneqq (2z+2)_{z\in \llbracket 1, k\rrbracket}$.

There is an even number of ``R'' and ``L'' in the connection sequence. Hence, due to Corollary \ref{remark:number_crossing} the subdiagram contains $\ell=(k+1)$ crossings. Consequently, the signature associated to the corresponding operator is equal to $(-1)^{k+1}$.

On the other hand, the recurrence relation \eqref{eq:relation_recu_RPA} introduces a $(-1)$ factor at each step, resulting in a $(-1)^{k}$ factor in the evaluation of ${\mathcal{F}}[\hat{\mathcal{O}}_k]$.

As a consequence, the $\textbf{Y}\textbf{Y}^\top \cdots \textbf{Y}\textbf{Y}^\top\textbf{X}\textbf{X}^\top$ matrix is present in the North-West block of $\tilde{{\gamma}}^\Delta_{k}$ with a negative sign.


More generally, for every $(\textbf{Y}\textbf{Y}^\top \cdots \textbf{Y}\textbf{Y}^\top\textbf{X}\textbf{X}^\top)_{r,s}$ contribution to the $r\times s$ element of $\tilde{\boldsymbol{\gamma}}^\Delta_{k}$, there is a sub-diagram which leads to that contribution. Corollary \ref{remark:occupancy_rs} tells us that the two extremities of the dash bend to the upper half-plane, and that the connection sequence will also contain an even number of ``$R$'' and of ``$L$''. Therefore, each $\textbf{Y}\textbf{Y}^\top \cdots \textbf{Y}\textbf{Y}^\top\textbf{X}\textbf{X}^\top$ contribution in the North-West block of the $\tilde{\boldsymbol{\gamma}}^\Delta_{k}$ matrix is affected by a negative factor. For that reason there cannot be two contributions, scalar multiple of $\textbf{Y}\textbf{Y}^\top \cdots \textbf{Y}\textbf{Y}^\top\textbf{X}\textbf{X}^\top$, that will cancel each other. 
\end{proof}

\begin{proposition}\label{prop:absenceXXYY}
For every natural integer $k$ strictly greater than one, the North-West — i.e., the occupied--occupied — block of $\tilde{\boldsymbol{\gamma}}^\Delta_{k}$ does not contain any scalar multiple of
\begin{equation}\label{eq:XXYY}
{\normalfont\textbf{X} \textbf{X}^\top\underbrace{\textbf{Y}\textbf{Y}^\top \cdots \textbf{Y}\textbf{Y}^\top}_{\textbf{Y}\textbf{Y}^\top \,\text{appears}\, k\,\text{times}}.}
\end{equation} 
\end{proposition}

\begin{proof}
Let $(r,s)$ be a couple of integers, both belonging to $\llbracket 1, N \rrbracket$. The $(\textbf{X} \textbf{X}^\top\textbf{Y}\textbf{Y}^\top \cdots \textbf{Y}\textbf{Y}^\top)_{r,s} $ matrix element can be written as a summation
\begin{equation}\label{eq:decompositionXXYYYY}
\hspace*{-.15cm}\sum_{l_1=1}^N\cdots\sum_{l_{k-1}=1}^{N}\sum_{c_1=N+1}^L \cdots \sum_{c_{k} = N+1}^{L}
(\textbf{X})_{r,c_1-N}^\wdag(\textbf{X}^\top)_{c_1-N,l_1}
(\textbf{Y})_{l_1,c_2-N}^\wdag(\textbf{Y}^\top)_{c_2-N,l_2}
\cdots
(\textbf{Y})_{l_{k-1},c_k-N}^\wdag(\textbf{Y}^\top)_{c_k-N,s}
\end{equation}
Notice that, in these conditions, the $r\times s$ element of a sub-matrix composing the occupied-occupied block of the total — i.e., $L \times L$ — matrix is also the $r\times s$ element of the total matrix.

Two matrix elements in the right-hand side of \eqref{eq:decompositionXXYYYY} share a common index only when the corresponding operators are \textit{contracted} in the chain expectation value — see Ref. \cite{etienne_comprehensive_2021}. In these conditions the two corresponding extremities in the diagrammatic simplification of the 1--CNC at stake are connected. Thus, two matrices follow each other in a term of the North-West block of $\tilde{{\gamma}}^\Delta_{k}$ if extremities of the diagrammatic representation of the associated operators are connected in the sub-diagram.

If there is a $(\textbf{X} \textbf{X}^\top\textbf{Y}\textbf{Y}^\top \cdots \textbf{Y}\textbf{Y}^\top)_{r,s}$ contribution to the $r \times s$ element of $\tilde{{\gamma}}^\Delta_{k}$, it is coming from a sub-diagram in which the extremity labeled with $\hat{s}^\dag$ is connected with the diagrammatic representation of an excitation operator associated with an $\textbf{y}$ component coefficient, and the extremity labeled with $\hat{r}$ is connected with the diagrammatic representation of a deexcitation operator associated with an $\textbf{x}$ component coefficient. The pattern is represented in Figure \ref{fig:pattern_YX}. Knowing that $r$ and $s$ can only take values in $\llbracket 1,N\rrbracket$ we have that the extremities of the dash in Figure \ref{fig:pattern_YX} must bend in the upper half-plane. Moreover, due to Proposition \ref{prop:nullity_criterion_sub_diagram} we also know that the two extremities of the sub-diagram must be in the upper half-plane for the contribution of the operator corresponding to the sub-diagram to be non-zero in the expectation value.

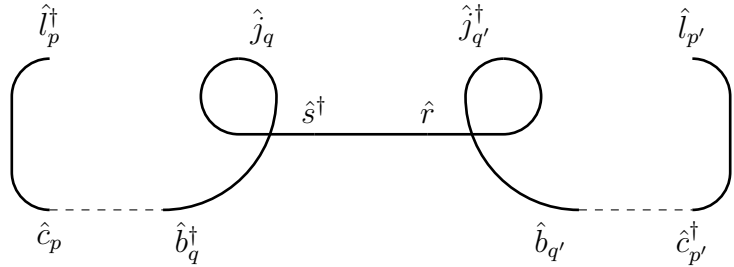
\begin{figure}[h!]
\centering
\begin{tikzpicture}
\pdeex{-2}{\hat{l}_p^\dag}{\hat{c}_p} 
\draw[dashed] (-1.5,-1) -- (0,-1);
\exci{\hat{j}_q}{\hat{b}_q^\dag}\adaga{deex}{\hat{s}^\dag}{exci}{\hat{r}}\deex{\hat{j}_{q'}^\dag}{\hat{b}_{q'}}
\draw[dashed] (0,-1) -- (1.5,-1);
\pexci{2}{\hat{l}_{p'}}{\hat{c}_{p'}^\dag}
\end{tikzpicture}
\caption{Sub-diagram which corresponds to an operator whose Fermi-vacuum expectation value will contribute to an element of a matrix with the \eqref{eq:XXYY} form in the occupied-occupied block of $\tilde{\bm{\gamma}}^\Delta_{k}$. Dashed lines represent a group of closing and opening bracket. The $p$, $p'$, $q$, and $q'$ values are such that $2k+2=\mathrm{max}(p,p')$ and $1=\mathrm{min}(q,q')$.}
\label{fig:pattern_YX}
\end{figure}

\noindent According to the nomenclature --- see Figure \ref{fig:nomenclature_bracket} ---, both diagrammatic representations connected to the extremities of the dash originate from $\mathcal{\hat{T}}^\dag$. However, we know that there is a $\mathcal{\hat{T}}$ operator in the innermost commutator, not a $\mathcal{\hat{T}}^\dag$. For that reason we know that a sub-diagram such as the one in Figure \ref{fig:pattern_YX} cannot be generated when diagrammatically evaluating $\mathcal{F}[\hat{\mathcal{O}}_k]$, hence no scalar multiple contribution of \eqref{eq:XXYY} can be encountered in the North-West block of $\tilde{\bm{\gamma}}^\Delta_k$.
\end{proof}
\newpage
\begin{proposition}
In the $(\tilde{\bm{\gamma}}^\Delta_k)_{k\in\mathbb{N}}^\wdag$ matrix sequence, only $\tilde{\bm{\gamma}}^\Delta_0$ is symmetric.
\end{proposition}
\begin{proof}
We saw with \eqref{eq:EOMisFA} that $\tilde{\bm{\gamma}}^\Delta_0$ is actually $\tilde{\bm{\gamma}}^\Delta$, a symmetric matrix whose expression is given in \eqref{eq:RPA_1DDM}. Proposition \ref{prop:gamma1nonsymmetric} showed that $\tilde{\bm{\gamma}}^\Delta_1$ is not symmetric. Combining propositions \ref{prop:existYYYYXX} and \ref{prop:absenceXXYY} shows that all the $\tilde{\bm{\gamma}}^\Delta_k$ matrices are non-symmetric when $k$ is strictly greater than one.
\end{proof}

\subsection{A conjecture}

\begin{conjecture}
Let ${\normalfont\textbf{T}}$ be the $L\times L$ matrix defined as
\begin{eqnarray}\label{eq:Tmatrixconjecture}
\normalfont\textbf{T} \coloneqq 
 \left(
  \begin{array}{cc}
   \textbf{0}_{N} & \textbf{X}          \\
   \textbf{Y}^\top & \textbf{0}_{L-N} \\
  \end{array}\right).
\end{eqnarray}
We postulate that in the $(\tilde{\bm{\gamma}}^\Delta_k)_{k\in\mathbb{N}}^\wdag$ matrix sequence each element is partitioned as
\begin{equation*}
\tilde{\boldsymbol{\gamma}}^\Delta_{k} = \boldsymbol{\Delta}^{\mathrm{NW}}_{k}\oplus\boldsymbol{\Delta}^{\mathrm{SE}}_{k}.
\end{equation*}
The matrix sequence itself is recursively defined as follows: $(\tilde{\bm{\gamma}}^\Delta_0 = \tilde{\bm{\gamma}}^\Delta)$, and
\begin{align*}
\normalfont\forall k \geq 1, \,
\tilde{\boldsymbol{\gamma}}^\Delta_{k+1}
= -\dfrac{1}{2}\left[\textbf{T}^\top,\left[ \tilde{\boldsymbol{\gamma}}^\Delta_{k},\textbf{T}  \right] \right] =  
\dfrac{1}{2}&\left( \normalfont
  \boldsymbol{\Delta}^{\mathrm{NW}}_{k}\textbf{X}\textbf{X}^\top 
+ \textbf{Y}\textbf{Y}^\top\boldsymbol{\Delta}^{\mathrm{NW}}_{k} 
- \textbf{X}\boldsymbol{\Delta}^{\mathrm{SE}}_{k}\textbf{X}^\top 
- \textbf{Y}\boldsymbol{\Delta}^{\mathrm{SE}}_{k}\textbf{Y}^\top 
\right)\\
\oplus 
&\left( \normalfont
  \textbf{X}^\top\textbf{X}\boldsymbol{\Delta}^{\mathrm{SE}}_{k} 
+ \boldsymbol{\Delta}^{\mathrm{SE}}_{k}\textbf{Y}^\top\textbf{Y}
- \textbf{X}^\top\boldsymbol{\Delta}^{\mathrm{NW}}_{k}\textbf{X}
- \textbf{Y}^\top\boldsymbol{\Delta}^{\mathrm{NW}}_{k}\textbf{Y} \right).
\end{align*}
\end{conjecture}
\noindent The conjecture has been verified using a homemade code of symbolic calculation, up to $k=3$. No element of this matrix sequence is symmetric except the first.

\section{Second approach — Recasting our case in a symmetrized version of the problem}
\subsection{Re-statement of the problem}
\begin{definition}[Symmetrized double commutator]
Let $\hat{A}$, $\hat{B}$, and $\hat{C}$ be three operators. The {\normalfont symmetrized double commutator} $[\cdot\,,\cdot\,,\cdot]$ is defined as
$$[\hat{A},\hat{B},\hat{C}] \coloneqq \dfrac{1}{2}\left([\hat{A},[\hat{B},\hat{C}]] + [[\hat{A},\hat{B}],\hat{C}]\right).$$
\end{definition}
\begin{definition}[Symmetrized exact elementary functional]\label{def:FnAsym}
Let $n$ be a natural integer belonging to $\llbracket 1,M\rrbracket$. The {\normalfont symmetrized exact elementary functional} corresponding to the $\Psi_0\longrightarrow\Psi_n$ transition is the linear form, denoted $F_n^S$, and defined as
\begin{align*}
\mathrm{End}(S_\mathbb{K}) &\longrightarrow \mathbb{K}\\
\hat{A}&\longmapsto F_n^S[\hat{A}]\coloneqq \braket{\Psi_0 |[\hat{T}_{0\rightarrow n}^\dag,\hat{A},\hat{T}_{0\rightarrow n}]|\Psi_0}.
\end{align*}
\end{definition}
\noindent The ``$S\,$'' superscript in the ``$F^S_n\,$'' expression denotes the symmetrization.
\begin{proposition}\label{prop:recu_EOM_sym}
Let $n$ be a natural integer belonging to $\llbracket 1, M\rrbracket$. For every $(r,s)$ couple of integers — both belonging to $\mathbb{N}^2$ — there exists an infinite sequence of operators, $(\hat{S}_k^\wdag)_{k\in\mathbb{N}}^\wdag$, defined as
\begin{equation}\label{def:relation_recu_sym}
\hat{S}_0 := \hat{s}^\dag\hat{r} \;\, \mathrm{and} \;\, \forall k\in\mathbb{N},\, \hat{S}_{k+1} \coloneqq -\dfrac{1}{2} \left[ \hat{{T}}_{0\rightarrow n}^\dag, \hat{S}_k,  \hat{{T}}_{0\rightarrow n} \right].
\end{equation}
such that
$$\forall k \in \mathbb{N},\, (\boldsymbol{\gamma}^\Delta_{0\rightarrow n})_{r,s} = F_n^S[\hat{S}_k].$$
\end{proposition}
\begin{proof}
The proof of Proposition \ref{prop:recu_EOM_sym} follows the exact same steps as the proof of {\normalfont Proposition \ref{prop:recu_EOM_simple}}.
\end{proof}
\begin{definition}[Symmetrized approached elementary functional]\label{def:FAsym}
The {\normalfont symmetrized approached elementary functional} is the linear form, denoted $\mathcal{F}^\mathcal{S}$, and defined as
\begin{align*}
\mathrm{End}(\mathcal{C}_\mathbb{R})&\longrightarrow \mathbb{R}\\
\hatcal{A}&\longmapsto\mathcal{F}^\mathcal{S}[\hatcal{A}]\coloneqq \braket{\psi_\star |[\hatcal{T}^\dag,\hatcal{A},\hatcal{T}]|\psi_\star}.
\end{align*}
\end{definition}
\noindent The ``$\mathcal{S}\, $'' superscript in the ``$\mathcal{F}^\mathcal{S}\,$'' expression denotes the symmetrization.
\begin{lemma}\label{lemma:A1dagA2A1}
Let $\hatcal{A}_1$ and $\hatcal{A}_2$ be two endomorphisms of $\mathcal{C}_\mathbb{R}$, with their adjoint being both also endomorphisms of $\mathcal{C}_\mathbb{R}$. Then,
$$[\hatcal{A}_1^\dag,\hatcal{A}_2^\wdag,\hatcal{A}_1]^\dag = [\hatcal{A}_1^\dag,\hatcal{A}_2^\dag,\hatcal{A}_1].$$
\end{lemma}
\begin{proof}
If $\hatcal{Q}_1^\wdag$ and $\hatcal{Q}_2^\wdag$ are two endomorphisms of $\mathcal{C}_\mathbb{R}$, with $\hatcal{Q}_1^\dag$ and $\hatcal{Q}_2^\dag$ being also two endomorphisms of $\mathcal{C}_\mathbb{R}$, we know that
$$\forall (\alpha,\beta)\in\mathbb{R}^2,\, (\alpha \hatcal{Q}_1+\beta\hatcal{Q}_2)^\dag = \alpha \hatcal{Q}_1^\dag+\beta\hatcal{Q}_2^\dag.$$
We also have $(\hatcal{Q}_1\hatcal{Q}_2)^\dag = \hatcal{Q}_2^\dag\hatcal{Q}_1^\dag$, and $[\hatcal{Q}_1,\hatcal{Q}_2]^\dag = [\hatcal{Q}_2^\dag,\hatcal{Q}_1^\dag]$. Applying those three identities to the case of $[\hatcal{A}_1^\dag,\hatcal{A}_2^\wdag,\hatcal{A}_1]^\dag$ immediately gives the desired result.
\end{proof}
\begin{proposition}\label{prop:FSAdagFSA}
Let $\hatcal{A}$ be an endomorphism of $\mathcal{C}_\mathbb{R}$, with $\hatcal{A}^\dag$ being also an endomorphism of $\mathcal{C}_\mathbb{R}$. Then,
$$\mathcal{F}^\mathcal{S}[\hatcal{A}^\dag] = \mathcal{F}^\mathcal{S}[\hatcal{A}].$$
\end{proposition}
\begin{proof}
The $\mathcal{F}^\mathcal{S}[\hatcal{A}]$ expression — see Definition \ref{def:FAsym} — alternatively reads
$$\mathcal{F}^\mathcal{S}[\hatcal{A}] = \braket{\psi_\star|[\hatcal{T}^\dag,\hatcal{A},\hatcal{T}]^\dag|\psi_\star}^*.$$
Knowing that $\hatcal{A}$ and $\hatcal{A}^\dag$, as well as $\hatcal{T}$ and $\hatcal{T}^\dag$ are all endomorphisms of $\mathcal{C}_\mathbb{R}$ — which solely contains real-valued wave functions and includes $\psi_\star$ — allows us to state that $\mathcal{F}^\mathcal{S}[\hatcal{A}]$ is real valued: 
$$\mathcal{F}^\mathcal{S}[\hatcal{A}] = \braket{\psi_\star|[\hatcal{T}^\dag,\hatcal{A},\hatcal{T}]^\dag|\psi_\star},$$
which, in virtue of Lemma \ref{lemma:A1dagA2A1}, is equal to
$$\mathcal{F}^\mathcal{S}[\hatcal{A}] = \braket{\psi_\star|[\hatcal{T}^\dag,\hatcal{A}^\dag,\hatcal{T}]|\psi_\star}$$
that is $\mathcal{F}^\mathcal{S}[\hatcal{A}^\dag]$.
\end{proof}
\noindent We now build $({}^\mathcal{S}\tilde{\boldsymbol{\gamma}}^\Delta_{k})_{k\in\mathbb{N}}^\wdag$, an infinite sequence of $L\times L$ matrices that are such that, for every $(r,s)$ couple of natural integers — both belonging to $\llbracket 1, L \rrbracket$ —, there is an infinite sequence of operators, $(\hatcal{S}_\ell^\wdag)^\wdag_{\ell\in\mathbb{N}}$, defined as
\begin{equation}\label{eq:relation_recu_RPA_sym}
\hatcal{S}_0 := \hat{s}^\dag\hat{r} \;\, \mathrm{and} \;\, \forall \ell \in \mathbb{N},\, \hatcal{S}_{\ell+1} \coloneqq -\dfrac{1}{2} \left[ \hat{\mathcal{T}}^\dag, \hatcal{S}_\ell,  \hat{\mathcal{T}} \right]
\end{equation}
such that
\begin{equation*}
\forall k \in \mathbb{N},\, ({}^\mathcal{S}\tilde{\boldsymbol{\gamma}}^\Delta_{k})_{r,s} \coloneqq \mathcal{F}^\mathcal{S}[\hatcal{S}_k].
\end{equation*}
\begin{lemma}\label{lemma:jacobi}
Let $\psi_\star$ be the function defined in \eqref{eq:def:psi0}, $\hatcal{T}$ the map defined in \eqref{eq:def:Tcurved}, and $(r,s)$ any couple of integers, both belonging to $\llbracket 1, L\rrbracket$. Then,
\begin{equation*}
\braket{\psi_\star|[\hatcal{T}^\dag,[\hat{s}^\dag\hat{r},\hatcal{T}]]|\psi_\star} = \braket{\psi_\star|[[\hatcal{T}^\dag,\hat{s}^\dag\hat{r}],\hatcal{T}]|\psi_\star}.
\end{equation*}
\end{lemma}
\begin{proof}
\noindent Let $\hat{A}$, $\hat{B}$, and $\hat{C}$ be three endomorphisms on a vector space $V$, and $\hat{\mathds{1}}_V$ be the identity map in $V$. We know that since commutators are Lie brackets, they must satisfy the Jacobi identity:
\begin{equation}\label{eq:Jacobi}
 [ \hat{A},  [ \hat{B}, \hat{C}  ]  ] + 
 [ \hat{C},  [ \hat{A}, \hat{B}  ]  ] + 
 [ \hat{B},  [ \hat{C}, \hat{A}  ]  ]
= {0}\hat{\mathds{1}}_V.
\end{equation}
Reorganizing \eqref{eq:Jacobi} gives
\begin{equation}\label{eq:Jacobi1}
  [ \hat{A},  [ \hat{B}, \hat{C}  ]  ] 
+   [ \hat{B},  [ \hat{C}, \hat{A}  ]  ]  
=  [  [ \hat{A}, \hat{B}  ], \hat{C}  ]  .
\end{equation}
If we take $(V=\mathcal{C}_\mathbb{R})$, $(\hat{A} = \hatcal{T}^\dag)$, $(\hat{B} = \hat{s}^\dag\hat{r})$, and $(\hat{C} = \hatcal{T})$, due to \eqref{eq:Jacobi1} it is then sufficient to prove that
\begin{equation*}
\mathscr{D}^{r,s}\coloneqq \braket{\psi_\star| [ \hat{s}^\dag\hat{r}, [ \hatcal{T}, \hatcal{T}^\dag  ]  ] |\psi_\star} = 0 
\end{equation*}
for proving our Lemma. We can develop $\mathscr{D}^{r,s}$ using four contributions:
\begin{align*}
\mathscr{D}^{r,s}_{\mathrm{I}}   &\coloneqq \braket{\psi_\star | \hat{s}^\dag\hat{r} \hatcal{T} \hatcal{T}^\dag |\psi_\star}, \\
\mathscr{D}^{r,s}_{\mathrm{II}}  &\coloneqq \braket{\psi_\star | \hat{s}^\dag\hat{r} \hatcal{T}^\dag \hatcal{T} |\psi_\star}, \\
\mathscr{D}^{r,s}_{\mathrm{III}} &\coloneqq \braket{\psi_\star | \hatcal{T} \hatcal{T}^\dag \hat{s}^\dag\hat{r} |\psi_\star}, \\
\mathscr{D}^{r,s}_{\mathrm{IV}}  &\coloneqq \braket{\psi_\star | \hatcal{T}^\dag \hatcal{T} \hat{s}^\dag\hat{r} |\psi_\star}.
\end{align*}
Then, $\mathscr{D}^{r,s}$ reads $(\mathscr{D}^{r,s} = \mathscr{D}^{r,s}_{\mathrm{I}} - \mathscr{D}^{r,s}_{\mathrm{II}} - \mathscr{D}^{r,s}_{\mathrm{III}} + \mathscr{D}^{r,s}_{\mathrm{IV}})$. We can develop $\mathscr{D}^{r,s}_{\mathrm{I}}$:
\begin{equation*}
\mathscr{D}^{r,s}_{\mathrm{I}} = 
\braket{\psi_\star | \hat{s}^\dag\hat{r} \hatcal{T}_y^\wdag \hatcal{T}^\dag_y |\psi_\star} - 
\braket{\psi_\star | \hat{s}^\dag\hat{r} \hatcal{T}_x^\wdag \hatcal{T}^\dag_y |\psi_\star},
\end{equation*}
that is to say, $\mathscr{D}^{r,s}_{\mathrm{I}}$ is equal to
\begin{align*}
  \sum_{i=1}^N\sum_{j=1}^N \sum_{a=N+1}^{L}\sum_{b=N+1}^{L} (({\textbf{Y}})_{i,a-N}({\textbf{Y}}^\top)_{b-N,j}
    \braket{\psi_\star| \hat{s}^\dag\hat{r} \hat{i}^\dag\hat{a} \hat{b}^\dag\hat{j}|\psi_\star} - ({\textbf{X}})_{i,a-N}({\textbf{Y}}^\top)_{b-N,j}
    \braket{\psi_\star| \hat{s}^\dag\hat{r} \hat{a}^\dag\hat{i} \hat{b}^\dag\hat{j}|\psi_\star}).
\end{align*}
Rules recalled in Ref. \cite{etienne_comprehensive_2021} give $\braket{\psi_\star| \hat{s}^\dag\hat{r} \hat{i}^\dag\hat{a} \hat{b}^\dag\hat{j}|\psi_\star} = 
\delta_{a,b} \delta_{i,j} \delta_{r,s} n_sn_r$ and $\braket{\psi_\star| \hat{s}^\dag\hat{r} \hat{a}^\dag\hat{i} \hat{b}^\dag\hat{j}|\psi_\star} = 0$, so that
\begin{align*}
\mathscr{D}^{r,s}_{\mathrm{I}} &= \sum_{i=1}^N\sum_{a=N+1}^{L} ({\textbf{Y}})_{i,a-N}({\textbf{Y}}^\top)_{a-N,i} \delta_{r,s}n_rn_s,\\
&= \vartheta_{y} \delta_{r,s}n_rn_s,
\end{align*}
where $\vartheta_y \coloneqq \mathrm{tr}(\textbf{YY}^\top)$. We can proceed in a similar fashion with $\mathscr{D}^{r,s}_{\mathrm{II}}$, $\mathscr{D}^{r,s}_{\mathrm{III}}$, and $\mathscr{D}^{r,s}_{\mathrm{IV}}$ and obtain
\begin{align*}
\mathscr{D}^{r,s}_{\mathrm{II}}  &= \vartheta_{x} \delta_{r,s}n_rn_s, \\
\mathscr{D}^{r,s}_{\mathrm{III}} &= \vartheta_{y} \delta_{r,s}n_rn_s, \\
\mathscr{D}^{r,s}_{\mathrm{IV}}  &= \vartheta_{x} \delta_{r,s}n_rn_s,
\end{align*}
where $\vartheta_x \coloneqq \mathrm{tr}(\textbf{XX}^\top)$. Introducing these results in $\mathscr{D}^{r,s}$ gives the desired result.
\end{proof}
\begin{proposition}\label{prop:FO0FSS0}
Let $(\hatcal{O}_\ell^\wdag)^\wdag_{\ell\in\mathbb{N}}$ be the sequence of operators defined in \eqref{eq:relation_recu_RPA}, and $(\hatcal{S}_\ell^\wdag)^\wdag_{\ell\in\mathbb{N}}$ be the sequence of operators defined in \eqref{eq:relation_recu_RPA_sym}. Let $\mathcal{F}$ (respectively, $\mathcal{F}^\mathcal{S}$) be the functional defined in {\normalfont Definition \ref{def:FA}} (respectively, {\normalfont Definition \ref{def:FAsym}}). Then,
$\mathcal{F}[\hatcal{O}_0] = \mathcal{F}^\mathcal{S}[\hatcal{S}_0]$.
\end{proposition}
\begin{proof}
Proof of Proposition \ref{prop:FO0FSS0} directly follows from Lemma \ref{lemma:jacobi} and the definition of the symmetrized double commutator.
\end{proof}
\noindent When dealing with the exact case, we saw with Proposition \ref{prop:recu_EOM_sym} that the choice for the argument of $F_n^S$ — as long as we use an element of the infinite sequence built in the proposition — is arbitrary. For building the \textit{reference} one-body reduced difference density matrix for the TDHF and TDDFRT methods we have shown with Proposition \ref{prop:FO0FSS0} that we can take a symmetrized version of \eqref{eq:EOM_1DDM} and replace $\Psi_0$ by $\psi_\star$, and $\hat{T}_{0\rightarrow n}$ by $\hatcal{T}$. That is to say, according to Proposition \ref{prop:FO0FSS0}, we actually see that each \textit{reference} one-body reduced difference density matrix element value for the TDHF and TDDFRT methods is the $\mathcal{F}^\mathcal{S}$ functional value with the first element of each $\hatcal{S}$ operator sequence built in \eqref{eq:relation_recu_RPA_sym} as argument for $\mathcal{F}^\mathcal{S}$. We wonder whether this choice is arbitrary, i.e., if replacing $\Psi_0$ by $\psi_\star$, and $\hat{T}_{0\rightarrow n}$ by $\hatcal{T}$ in other functionals than $F_n^S[\hat{O}_0]$ — which is equivalent to selecting other possible choices for the argument of $\mathcal{F}^\mathcal{S}$ — would lead to matrices that one cannot reject based on physical criteria. This would question — again — the uniqueness of the representation of TDHF and/or TDDFRT molecular electronic transitions.
\begin{proposition}
All the elements of the $({}^\mathcal{S}\tilde{\boldsymbol{\gamma}}^\Delta_{k})_{k\in\mathbb{N}}^\wdag$ sequence are $L\times L$ symmetric matrices.
\end{proposition}
\begin{proof}
For this proof we need to make explicit the $(r,s)$--dependence of the $\hatcal{S}$ operators in \eqref{eq:relation_recu_RPA_sym}:
$${}^{rs}\hatcal{S}_0^\wdag\coloneqq \hat{s}^\dag\hat{r}\;\mathrm{and}\; \forall \ell\in\mathbb{N},\,{}^{rs}\hatcal{S}_{\ell+1}^\wdag\coloneqq -\dfrac{1}{2}\left[\hatcal{T}^\dag,{}^{rs}\hatcal{S}_\ell^\wdag,\hatcal{T}\right].$$
A recursive application of Lemma \ref{lemma:A1dagA2A1} allows to write
\begin{equation}\label{eq:SkrsdagSksr}
\forall k\in\mathbb{N},\, \forall (r,s)\in\llbracket 1,L\rrbracket^2, {}^{rs}\hatcal{S}_k^\dag = {}^{sr}\hatcal{S}_k^\wdag.
\end{equation}
Then, knowing that
$$\forall k \in \mathbb{N},\,\forall (r,s)\in\llbracket 1,L\rrbracket^2,\, ({}^\mathcal{S}\tilde{\bm{\gamma}}^\Delta_k)_{r,s}\coloneqq\mathcal{F}^\mathcal{S}[{}^{rs}\hatcal{S}_k^\wdag],$$
and applying Proposition \ref{prop:FSAdagFSA} together with \eqref{eq:SkrsdagSksr} allows to write
$$\forall k \in \mathbb{N},\,\forall (r,s)\in\llbracket 1,L\rrbracket^2,\, ({}^\mathcal{S}\tilde{\bm{\gamma}}^\Delta_k)_{r,s}=\mathcal{F}^\mathcal{S}[{}^{sr}\hatcal{S}_k^\wdag],$$
i.e.,
$$\forall k \in \mathbb{N},\,\forall (r,s)\in\llbracket 1,L\rrbracket^2,\, ({}^\mathcal{S}\tilde{\bm{\gamma}}^\Delta_k)_{r,s}= ({}^\mathcal{S}\tilde{\bm{\gamma}}^\Delta_k)_{s,r},$$
that is to say 
$$\forall k \in \mathbb{N},\,({}^\mathcal{S}\tilde{\bm{\gamma}}^\Delta_k)^\top = {}^\mathcal{S}\tilde{\bm{\gamma}}^\Delta_k,$$
which is the desired result.
\end{proof}
\noindent We also provide an alternative proof that is in line with the previous derivations for the non-symmetrized double commutators.
\begin{proof}
We have seen with Proposition \ref{prop:FO0FSS0} that ${}^\mathcal{S}\tilde{\boldsymbol{\gamma}}^\Delta_{0}$ and $\tilde{\boldsymbol{\gamma}}^\Delta_{0}$ are identical — and symmetric. Let $k$ be a natural integer superior or equal to one. Consider a $(2k+2)$-list of matrices, $(\textbf{Z}_1, \ldots, \textbf{Z}_{2k+2})$, each of them being an $\textbf{X}$ or a $\textbf{Y}$ matrix.  We show below that if, in the North-West block of ${}^\mathcal{S}\tilde{\bm{\gamma}}_k^\Delta$, there is a term that is a real multiple of the
\begin{equation}\label{eq:occoccmat}
\textbf{Z}_1^\wdag \textbf{Z}_2^\top\cdots \textbf{Z}_{2k+1}^\wdag \textbf{Z}_{2k+2}^\top
\end{equation}
matrix, then there is also another term that is its symmetric: Let $(r,s)$ be a couple of natural integers, both belonging to $\llbracket 1, N \rrbracket$. Consider the $(\textbf{Z}_1^\wdag \textbf{Z}_2^\top\cdots \textbf{Z}_{2k+1}^\wdag \textbf{Z}_{2k+2}^\top)_{r,s}$ matrix element, i.e.,
\begin{equation*}
 \sum_{l_1=1}^N \cdots \sum_{l_{k-1}=1}^{N}\sum_{c_1=N+1}^L \cdots \sum_{c_{k} = N+1}^{L}
(\textbf{Z}_1)_{r,c_1-N}^\wdag
(\textbf{Z}_2^\top)_{c_1-N,l_1}
\cdots
(\textbf{Z}_{2k-1})_{l_{k-1},c_k-N}^\wdag
(\textbf{Z}_{2k}^\top)_{c_k-N,s}.
\end{equation*}
Its evaluation actually involves an $S$ sub-diagram described by a connection sequence $C$ and an operator sequence $O$. 
Consider the $C'$ connection sequence obtained by transforming each ``$L$'' (respectively, ``$R$'') in $C$ by an ``$R$'' (respectively, ``$L$'').
Consider also the $O'$ operator sequence obtained from the transformation described in Table \ref{tab:OOprime}.

\begin{table}[h!]
\begin{center}
\begin{tabular}{l|cccc}
$m^{\text{th}}$ operator in $O$ & $\hat{E}_{i_m}^{a_m}$ & $\hat{D}_{a_m}^{i_m}$ & $\hat{E}_{j_m}^{b_m}$ & $\hat{D}_{b_m}^{j_m}$ \\ 
\hline
$m^{\text{th}}$ operator in $O'$& $\hat{D}_{b_m}^{j_m}$ & $\hat{E}_{j_m}^{b_m}$ & $\hat{D}_{a_m}^{i_m}$ & $\hat{E}_{i_m}^{a_m}$    \\ 
\end{tabular} 
\caption{Operator transformation from the $O$ operator sequence to the $O'$ operator sequence. If an operator from the top right of the table is met at the $m^\text{th}$ place in $O$, it is substituted by the corresponding operator in the bottom right of the table in $O'$. The $m$ index value ranges from $1$ to $2k+2$.}
\label{tab:OOprime}
\end{center}
\end{table}

\noindent In this way the $m^\text{th}$ operator in $O$ and $O'$ are both associated with the same matrix contribution, i.e., an $\textbf{X}$ matrix contribution or a $\textbf{Y}$ matrix contribution — see Figure \ref{fig:nomenclature_bracket}.

We now consider the $S'$ sub-diagram generated by the $C'$ connection sequence and the $O'$ connection sequence. If the diagrammatic representation of the $m^\text{th}$ element of $O$ is the $l^\text{th}$ bracket on the right (respectively, left) of the dash in $S$, then the diagrammatic representation of the $m^\text{th}$ element of $O'$ in $S'$ is the $l^\text{th}$ bracket that is at the left (respectively, right) of the dash. A direct consequence of this is that the operator corresponding to the $S'$ sub-diagram is obtained by exchanging all the $r$ and $s$ symbols and, for every value of $m$, all the $i_m$ and $j_m$ symbols, and all the $a_m$ and $b_m$ symbols -- for operators and matrix indices in the Hermitian conjugate of the operator corresponding to the $S$ sub-diagram, leading to the $(\textbf{Z}_{2k+2}^\wdag \textbf{Z}_{2k+1}^\top \cdots \textbf{Z}_2^\wdag\textbf{Z}_1^\top)_{r,s}$ matrix element contribution, i.e.,
\begin{equation*}
\sum_{l_1=1}^N \cdots \sum_{l_{k-1}=1}^{N}\sum_{c_1=N+1}^L\cdots \sum_{c_{k} = N+1}^{L}
(\textbf{Z}_{2k+2})_{r,c_1-N}^\wdag
(\textbf{Z}_{2k+1})_{c_1-N,l_1}^\top
\ldots
(\textbf{Z}_2)_{l_{k-1},c_k-N}^\wdag
(\textbf{Z}_1)_{c_k-N,s}^\top.
\end{equation*}
This contribution is the symmetric of the contribution in \eqref{eq:occoccmat}. Hence,
\begin{equation*}
\left[\textbf{Z}_1^\wdag \textbf{Z}_2^\top\cdots \textbf{Z}_{2k+1}^\wdag \textbf{Z}_{2k+2}^\top\right]^\top = \textbf{Z}_{2k+2}^\wdag \textbf{Z}_{2k+1}^\top \cdots \textbf{Z}_2^\wdag\textbf{Z}_1^\top.
\end{equation*}
Since there will always be an $S'$ that can be built corresponding to a given $S$, the two resulting contributions will always have the same coefficient in ${}^\mathcal{S}\tilde{\bm{\gamma}}_k^\Delta$.

The same reasoning holds for the South-East block, which is then also symmetric. Finally, due to Corollary \ref{remark:occupancy_rs}, the anti-diagonal blocks are zero matrices. This concludes the proof.
\end{proof}
\noindent We provide here the expression of ${}^\mathcal{S}\tilde{\boldsymbol{\gamma}}^\Delta_{1}$, that can be compared with $\tilde{\bm{\gamma}}^\Delta_1$ given in \eqref{eq:gamma_1}:
\begin{align*}
{}^\mathcal{S}\tilde{\boldsymbol{\gamma}}^\Delta_{1} =
&\dfrac{1}{8}\left(
-8 \textbf{X}\textbf{X}^\top\textbf{X}\textbf{X}^\top
-4 \textbf{X}\textbf{X}^\top\textbf{Y}\textbf{Y}^\top
-4 \textbf{X}\textbf{Y}^\top\textbf{Y}\textbf{X}^\top
-4 \textbf{Y}\textbf{X}^\top\textbf{X}\textbf{Y}^\top
-4 \textbf{Y}\textbf{Y}^\top\textbf{X}\textbf{X}^\top
-8 \textbf{Y}\textbf{Y}^\top\textbf{Y}\textbf{Y}^\top
\right)\\
&\oplus \left(
 8 \textbf{X}^\top\textbf{X}\textbf{X}^\top\textbf{X}
+4 \textbf{X}^\top\textbf{X}\textbf{Y}^\top\textbf{Y}
+4 \textbf{X}^\top\textbf{Y}\textbf{Y}^\top\textbf{X}
+4 \textbf{Y}^\top\textbf{X}\textbf{X}^\top\textbf{Y}
+4 \textbf{Y}^\top\textbf{Y}\textbf{X}^\top\textbf{X}
+8 \textbf{Y}^\top\textbf{Y}\textbf{Y}^\top\textbf{Y}
\right).
\end{align*}

\subsection{Special case — The Tamm-Dancoff limit}

\begin{proposition}\label{prop:gammaDeltaTDA0}
The ${}^\mathcal{S}\tilde{\bm{\gamma}}^\Delta_0$ matrix in the ${\normalfont\textbf{Y}}\rightarrow {\normalfont\textbf{0}}_{N\times(L-N)}$ limit is the $\tilde{\bm{\gamma}}^\Delta$ matrix from \eqref{eq:RPA_1DDM} in the ${\normalfont\textbf{Y}}\rightarrow {\normalfont\textbf{0}}_{N\times(L-N)}$ limit, i.e.,
\begin{equation*}
{\normalfont\left(- 
\textbf{X}\textbf{X}^\top
\right) 
\oplus  
\textbf{X}^\top\textbf{X}
 }.
\end{equation*}
\end{proposition}
\begin{proof}
Considering \eqref{eq:EOMisFA}, together with \eqref{eq:RPA_1DDM} and \eqref{eq:relation_recu_RPA}, and Proposition \ref{prop:FO0FSS0} gives the desired result.
\end{proof}
\begin{proposition}\label{prop:gammaDeltaTDA}
Let $k$ be a natural integer different from zero. The ${}^\mathcal{S}\tilde{\bm{\gamma}}^\Delta_k$ matrix in the ${\normalfont\textbf{Y}}\rightarrow {\normalfont\textbf{0}}_{N\times(L-N)}$ limit is
\begin{equation*}
{\normalfont\left(-\left(
\textbf{X}\textbf{X}^\top
\right)^{k+1}\right)
\oplus \left(
\textbf{X}^\top\textbf{X}
\right)^{k+1}}.
\end{equation*}
\end{proposition}
\begin{proof}
According to the recurrence relation {\normalfont \eqref{def:relation_recu_sym}}, the 
\begin{equation*}
\left[ \hat{\mathcal{T}}^\dag_x, \hatcal{S}_k,  \hat{\mathcal{T}}^\wdag_x \right]
\end{equation*}
symmetrized double commutator can be decomposed into a sum of $2^{2k+2}$ nested commutators involving $(k+1)$ times the $\mathcal{\hat{T}}_x$ operator and $(k+1)$ times the $\mathcal{\hat{T}}^\dag_x$ operator. A coefficient $2^{-2k-1}$ appears in front of each nested commutator. Each nested commutator is associated with a unique $(2k+2)$--long operator sequence. Therefore, any sub-diagram obtained in the diagrammatic development of any nested commutator is described by a $(2k+2)$--long connection sequence. Each entry of a connection sequence being an ``$L$'' or an ``$R$'', we have a total of $2^{2k+2}$ possible connection sequences. 

Combining {\normalfont Proposition \ref{prop:non_nullity}} with our nomenclature choice — see Figure \ref{fig:nomenclature_bracket} —, for a given connection sequence there exists a unique operator sequence that generates, together with the connection sequence, a sub-diagram corresponding to an operator which has an expectation value respectively to the Fermi vacuum that is susceptible of being non-zero. The number of sub-diagrams fulfilling these conditions is $2^{2k+2}$.

Among the $2^{2k+2}$ possible connection sequences, $2^{2k+1}$ contain an odd number of ``R'' entries and $2^{2k+1}$ contain an even number of ``R'' entries. 

\noindent In order to respect each block dimension, we know that non-zero matrix contributions in the North-West block must have the
\begin{equation*}
\left(\textbf{X}\textbf{X}^\top\right)^{k+1}
\end{equation*}
shape, and that non-zero matrix contributions in the South-East block must have the
\begin{equation*}
\left(\textbf{X}^\top\textbf{X}\right)^{k+1}
\end{equation*}
shape. According to {\normalfont Corollary \ref{remark:occupancy_rs}}, there are $2^{2k+1}$ identical matrix contributions in the North-West block and $2^{2k+1}$ identical matrix contributions in the South-East block.

The number of identical matrix contributions in a block ($2^{2k+1}$) cancels out the coefficient in front of the matrix ($2^{-2k-1}$, \textit{vide supra}). We also know, due to Corollary \ref{remark:occupancy_rs}, that the anti-diagonal blocks are zero matrices. Finally, following the same reasoning as in the proof of Proposition \ref{prop:existYYYYXX}, we can involve Corollary \ref{remark:number_crossing} and recurrence relation {\normalfont \eqref{def:relation_recu_sym}} for providing the final expression — with a ($-$) sign in the North-West block, and a ($+$) sign in the South-East block. This concludes the proof.
\end{proof}
\begin{lemma}\label{lemma:AspectrumAk}
Let {\normalfont \textbf{A}} be an $L\times L$ Hermitian matrix, and ${\normalfont\sigma(\textbf{A}) = \{\lambda_s(\textbf{A})\,:\, s\in \llbracket 1,L\rrbracket\}}$ its spectrum. Then,
$${\normalfont\forall k \in \mathbb{N}^*, \, \forall \lambda \in \sigma(\textbf{A}),\, \lambda^k\in \sigma(\textbf{A}^k).}$$
\end{lemma}
\begin{lemma}\label{lemma:detailedspectrumTDA}
Let $k$ be a natural integer. Let $({}^k\lambda^\uparrow_s)^\wdag_{s\in\llbracket 1,L\rrbracket}$ be the list of eigenvalues of ${}^\mathcal{S}\tilde{\bm{\gamma}}^\Delta_k$ in the ${\normalfont\textbf{Y}}\rightarrow {\normalfont\textbf{0}}_{N\times(L-N)}$ limit, sorted in the increasing order. Then,
$$(\forall s \in \llbracket 1, N \rrbracket, \, -1 \leq {}^k\lambda_{s}^\uparrow \leq 0) \, \land \, (\forall s \in \llbracket (N+1), L \rrbracket, \,  0 \leq {}^k\lambda_{s}^\uparrow \leq 1).$$
\end{lemma}
\begin{proof}
We know that $\left(\textbf{X}\textbf{X}^\top\right)$ is positive semidefinite since it is the matrix product of a matrix and its transpose. Since $(\textbf{x}\tvec\textbf{x} = 1)$ in the Tamm-Dancoff limit, we also know that the sum of the eigenvalues of $\left(\textbf{X}\textbf{X}^\top\right)$ is equal to one. For the same reasons the $\left(\textbf{X}^\top\textbf{X}\right)$ matrix is also positive semidefinite, with the sum of its eigenvalues being equal to one. Combining Lemma \ref{lemma:AspectrumAk} and Proposition \ref{prop:gammaDeltaTDA} then gives the desired result.
\end{proof}
\begin{proposition}
If more than one component of ${\normalfont\textbf{x}}$ in $\hatcal{T}$ is different from zero, then all the eigenvalues of ${}^\mathcal{S}\tilde{\bm{\gamma}}^\Delta_k$ in the ${\normalfont\textbf{Y}}\rightarrow {\normalfont\textbf{0}}_{N\times(L-N)}$ limit tend to zero when $k$ tends to infinity.
\end{proposition}
\begin{proof}
Knowing that in the Tamm-Dancoff limit $(\textbf{x}\tvec\textbf{x} = 1)$, the proposition is a direct consequence of Lemma \ref{lemma:detailedspectrumTDA} and Lemma \ref{lemma:AspectrumAk}.
\end{proof}
\noindent Endowed with Lemma \ref{lemma:detailedspectrumTDA} we know that no eigenvalue of any element of $({}^\mathcal{S}\tilde{\bm{\gamma}}^\Delta_k)_{k\in\mathbb{N}}$ is lower than $(-1)$ or greater than one. Moreover, propositions \ref{prop:gammaDeltaTDA0} and \ref{prop:gammaDeltaTDA} tell us that (i) each element of $({}^\mathcal{S}\tilde{\bm{\gamma}}^\Delta_k)_{k\in\mathbb{N}}$ is symmetric, and (ii) the trace of each element of $({}^\mathcal{S}\tilde{\bm{\gamma}}^\Delta_k)_{k\in\mathbb{N}}$ is zero.
\subsubsection{The case of TDHF}\label{subsub:TDHFTDAdiscussion}
The TDHF one-body reduced difference density matrix in the Tamm-Dancoff limit is known: It is the Configuration Interaction Singles (CIS) one-body reduced difference density matrix — see Ref. \cite{etienne_comprehensive_2021}. It is ${\normalfont\left(-\textbf{X}\textbf{X}^\top\right)\oplus\textbf{X}^\top\textbf{X}}$. In that case $\psi_\star$ is the Hartree-Fock wave function — denoted $\psi^\text{HF}$ in the following — and seen as the CIS approximation of the ground state wave function while $\psi^\text{CIS}=\hatcal{T}_x\psi^\text{HF}$ is the CIS approximation of an excited state wave function \cite{dreuw_single-reference_2005}, so that
$$\forall (r,s)\in\llbracket 1,L\rrbracket^2,\,\braket{\psi^\text{CIS}|\hat{s}^\dag\hat{r}|\psi^\text{CIS}} - \braket{\psi^\text{HF}|\hat{s}^\dag\hat{r}|\psi^\text{HF}} = \left({\normalfont\left(-\textbf{X}\textbf{X}^\top\right)\oplus\textbf{X}^\top\textbf{X}}\right)_{r,s}.$$
We know with Proposition \ref{prop:gammaDeltaTDA} that, unless exactly one component of $\textbf{x}$ is different from zero, no element in the $({}^\mathcal{S}\tilde{\bm{\gamma}}^\Delta_k)^\wdag_{k\in\mathbb{N}}$ matrix sequence — except the first — will give the correct TDHF one-body reduced difference density matrix when $\textbf{Y}$ will tend to zero. This lack of universality leads us to conclude that we cannot state about any element of the $({}^\mathcal{S}\tilde{\bm{\gamma}}^\Delta_k)^\wdag_{k\in\mathbb{N}}$ matrix sequence except the first that it possesses the required properties for being considered as a TDHF approximation to the one-body reduced difference density matrix of the system.
\subsubsection{The case of TDDFRT}
In the case of TDDFRT in the Tamm-Dancoff limit (TDA--TDDFRT) we are not aware of a universal consensus regarding whether $\psi_\star$ — here the Kohn-Sham wave function, denoted $\psi^\text{KS}$ in the following — can be seen as the TDA--TDDFRT approximation of the ground-state wave function while $\psi_\textbf{x}\coloneqq \hatcal{T}_x\psi^\text{KS}$ can be seen as the TDA--TDDFRT approximation of an excited-state wave function — see for example Figure 1 and section 3.3 of Ref. \cite{dreuw_single-reference_2005}. For instance, one could interpret $\psi_\textbf{x}$ as the TDDFRT--TDA approximation to one excited-state wave function that is unknowable, without admitting $\psi^\text{KS}$ as the TDA--TDDFRT ground-state wave function. Another possibility is to reject the possibility that there exists a (knowable) TDA--TDDFRT excited-state wave function, and to interpret $\psi_\textbf{x}$ as an auxiliary wave function which, when substituted into the expression of reduced quantities — Green's functions, density matrices, density functions, etc. —, or only in some of them, gives the TDA--TDDFRT approximation to these reduced quantities. Another possibility is to give $\psi_\textbf{x}$ no interpretation. We ignore whether the constitution, structure or derivation of the TDA--TDDFRT central problem may allow to determine unequivocally the interpretation of $\psi_\textbf{x}$. Hence, due to the possibility of a non-universal interpretation of $\psi_\textbf{x}$ (and $\psi^\text{KS}$) in the TDA--TDDFRT framework, we feel unable to directly transfer here the interpretation and conclusions we have drawn in the previous paragraph when TDHF was concerned.
\subsection{{Another conjecture}}
\begin{conjecture}
Let ${\normalfont\textbf{T}}$ be the $L\times L$ matrix defined in \eqref{eq:Tmatrixconjecture}. We postulate that in the $\normalfont({}^\mathcal{S}\tilde{\bm{\gamma}}^\Delta_k)_{k\in\mathbb{N}}^\wdag$ matrix sequence each element is partitioned as
\begin{equation*}
\normalfont{}^\mathcal{S}\tilde{\boldsymbol{\gamma}}^\Delta_{k} = {}^\mathcal{S}\!\boldsymbol{\Delta}^{\mathrm{NW}}_{k}\oplus{}^\mathcal{S}\!\boldsymbol{\Delta}^{\mathrm{SE}}_{k}.
\end{equation*}
The matrix sequence itself is recursively defined as follows: $({}^\mathcal{S}\tilde{\bm{\gamma}}^\Delta_0 = \tilde{\bm{\gamma}}^\Delta)$, and
\begin{equation*}
\normalfont\forall k \geq 1, \,
{}^\mathcal{S}\tilde{\boldsymbol{\gamma}}^\Delta_{k+1}
= -\dfrac{1}{2}
\left[\textbf{T}^\top,{}^\mathcal{S}\tilde{\boldsymbol{\gamma}}^\Delta_{k}, \textbf{T} \right],
\end{equation*}
i.e.,
\begin{align*}
\normalfont{}^\mathcal{S}\tilde{\boldsymbol{\gamma}}^\Delta_{k+1} = 
\dfrac{1}{4}&
\left( \normalfont
   \textbf{X}\textbf{X}^\top{}^\mathcal{S}\!\boldsymbol{\Delta}^{\mathrm{NW}}_{k}
+  \textbf{Y}\textbf{Y}^\top{}^\mathcal{S}\!\boldsymbol{\Delta}^{\mathrm{NW}}_{k}
+  {}^\mathcal{S}\!\boldsymbol{\Delta}^{\mathrm{NW}}_{k}\textbf{X}\textbf{X}^\top
+  {}^\mathcal{S}\!\boldsymbol{\Delta}^{\mathrm{NW}}_{k}\textbf{Y}\textbf{Y}^\top
-2 \textbf{X}{}^\mathcal{S}\!\boldsymbol{\Delta}^{\mathrm{SE}}_{k}\textbf{X}^\top 
-2 \textbf{Y}{}^\mathcal{S}\!\boldsymbol{\Delta}^{\mathrm{SE}}_{k}\textbf{Y}^\top 
\right)\\
\oplus 
&\left( \normalfont
  \textbf{X}^\top\textbf{X}{}^\mathcal{S}\!\boldsymbol{\Delta}^{\mathrm{SE}}_{k} 
    +  {}^\mathcal{S}\!\boldsymbol{\Delta}^{\mathrm{SE}}_{k} \textbf{X}^\top\textbf{X}
+ {}^\mathcal{S}\!\boldsymbol{\Delta}^{\mathrm{SE}}_{k}\textbf{Y}^\top\textbf{Y}
+ \textbf{Y}^\top\textbf{Y}{}^\mathcal{S}\!\boldsymbol{\Delta}^{\mathrm{SE}}_{k}
-2 \textbf{X}^\top{}^\mathcal{S}\!\boldsymbol{\Delta}^{\mathrm{NW}}_{k}\textbf{X}
-2 \textbf{Y}^\top{}^\mathcal{S}\!\boldsymbol{\Delta}^{\mathrm{NW}}_{k}\textbf{Y} \right).
\end{align*}
\end{conjecture}
\noindent All the elements of the matrix sequence in the conjecture are symmetric.

\section{Third approach — Recasting our problem into a more general framework: the equation-of-motion framework}

For a thorough and extensive review of the equations-of-motion framework, see Ref. \cite{herman1981analysis}.

\subsection{Generalities}

We define
$$T \coloneqq \left\lbrace \ket{\Psi_i}\bra{\Psi_j} \, : \, \forall (i,j)\in \llbracket 0,M\rrbracket ^2, \, (\Psi_i,\Psi_j)\in S^2\right\rbrace = \left\lbrace \hat{T}_i \, : \, i \in\llbracket 1, (M+1)^2\rrbracket\right\rbrace,$$
and $T_\mathbb{K} \coloneqq \mathrm{span}_\mathbb{K} T.$
\begin{lemma}\label{lemma:TitoAmap}
Let $\hat{A}$ be a linear map on $S_\mathbb{K}$. The $\hat{A}$ map is an endomorphism of $S_\mathbb{K}$ if and only if it is an element of $T_\mathbb{K}$:
$$\hat{A}\in\mathrm{End}(S_\mathbb{K})\Longleftrightarrow \hat{A}\in T_\mathbb{K}.$$
\end{lemma}
\begin{proof}
We first prove $(\hat{A}\in\mathrm{End}(S_\mathbb{K})\Longleftarrow \hat{A}\in T_\mathbb{K})$. The elements of $T$ are obviously endomorphisms of $S_\mathbb{K}$. We know that any linear combination of endomorphisms of a vector space $V$ is an endomorphism of $V$. Therefore, if $\hat{A}$ belongs to $T_\mathbb{K}$ it is an endomorphism of $S_\mathbb{K}$.

We now prove $(\hat{A}\in\mathrm{End}(S_\mathbb{K})\Longrightarrow \hat{A}\in T_\mathbb{K})$. If the $\hat{A}$ map belongs to $\mathrm{End}(S_\mathbb{K})$, it can be rewritten as
$$\hat{A} = \hat{\mathds{1}}_{S_\mathbb{K}} \hat{A} \hat{\mathds{1}}_{S_\mathbb{K}},$$
where $\hat{\mathds{1}}_{S_\mathbb{K}}$ is the identity map on $S_\mathbb{K}$. Since $S$ is a basis of $S_\mathbb{K}$, and since $S$ is orthonormal, the resolution of identity in $S_\mathbb{K}$ reads
$$\hat{\mathds{1}}_{S_\mathbb{K}} = \sum _{i=0}^M \ket{\Psi_i}\bra{\Psi_i}.$$
For that reason, we can rewrite $\hat{A}$ as
$$\hat{A} = \sum _{i=0}^M\sum_{j=0}^M \braket{\Psi_i |\hat{A}|\Psi_j} \ket{\Psi_i}\bra{\Psi_j},$$
which is an element of $T_\mathbb{K}$.
\end{proof}
\subsection{Transition properties with the EOM formalism}
\noindent We can order the eigenvalues of $\hat{H}$ corresponding to the elements of $S$ into a list: $(E_i)_{i\in\llbracket 0, M\rrbracket}$. The problem of finding, for a given $n$ in $\llbracket 1, M\rrbracket$ an operator, $\hat{O}^\dag_n$, such that for every $i$ in $\llbracket 1, (M+1)^2\rrbracket$ we have
\begin{equation}\label{eq:EOMprob_BASIC}
\braket{\Psi_{0}|[\hat{T}_i^\dag,[\hat{H},\hat{O}_n^\dag]]|\Psi_{0}} 
= (E_n-E_0)\braket{\Psi_{0}|[\hat{T}_i^\dag,\hat{O}_n^\dag]|\Psi_{0}}
\end{equation}
or, in a symmetrized version,
\begin{equation}\label{eq:EOMprob_SYMMETRIZED}
\braket{\Psi_{0}|[\hat{T}_i^\dag,\hat{H},\hat{O}_n^\dag]|\Psi_{0}} 
= (E_n-E_0)\braket{\Psi_{0}|[\hat{T}_i^\dag,\hat{O}_n^\dag]|\Psi_{0}},
\end{equation}
is known as the problem of solving the so-called \textit{equations-of-motion} for $\hat{H}$. Elements of 
$$\{ \hat{T}_{0\rightarrow m} \, : \, m \in \llbracket 1,M\rrbracket\}$$
are operator solutions of both \eqref{eq:EOMprob_BASIC} and \eqref{eq:EOMprob_SYMMETRIZED}, but it has been proved that the most general operator solution to both \eqref{eq:EOMprob_BASIC} and \eqref{eq:EOMprob_SYMMETRIZED} is a special type of element of $T_\mathbb{K}$ of the form
\begin{equation}\label{eq:def:Odagn}
{\hat{O}}^\dag_n :=
\ket{\Psi_{n}}\!\bra{\Psi_{0}} 
+ \alpha_{0,0} \ket{\Psi_{0}}\!\bra{\Psi_{0}} 
+ \sum_{p=1}^M\sum_{q=1}^M \alpha_{p,q}\ket{\Psi_{p}}\!\bra{\Psi_{q}}
\end{equation}
where the $\alpha$ scalars are arbitrary coefficients \cite{herman1981analysis,lynch_excited_nodate}. From this most general operator it is possible to express the transition expectation value of any $\hat{A}$ operator — assumed here to belong to $\mathrm{End}(S_\mathbb{K})$ — relatively to any couple of non-necessarily distinct excited states \cite{lynch_excited_nodate}:
\begin{equation}\label{eq:APsinPsim}
\forall (n,m)\in \llbracket 1,M\rrbracket^2,\,\bra{\Psi_n}\hat{A}\ket{\Psi_m} = \braket{\Psi_{0}|[{\hat{O}_{n}^\wdag},\hat{A},{\hat{O}}_m^\dag]|\Psi_{0}}
+\dfrac{1}{2}\braket{\Psi_{0}|[[{\hat{O}}_{n}^\wdag,{\hat{O}}_{m}^\dag ],\hat{A}]_+|\Psi_{0}}
\end{equation}
where the anticommutator ``$[\cdot,\cdot]_+$'' is defined as
\begin{equation*}
[\hat{A},\hat{B}] _+
\coloneqq \hat{A}\hat{B} + \hat{B}\hat{A}.
\end{equation*}
In \eqref{eq:APsinPsim}, when $n$ is equal to $m$, the relation reduces to
\begin{equation}\label{eq:lynch24modified}
\bra{\Psi_n}\hat{A}\ket{\Psi_n} = \braket{\Psi_{0}|[{\hat{O}_{n}^\wdag},\hat{A},{\hat{O}}_n^\dag]|\Psi_{0}}
+\dfrac{1}{2}\braket{\Psi_{0}|[[{\hat{O}}_{n}^\wdag,{\hat{O}}_{n}^\dag ],\hat{A}]_+|\Psi_{0}}.
\end{equation}
\subsection{The case of TDHF}
Due to Lemma \ref{lemma:TitoAmap} and the bilinearity of commutators, if $\hat{O}_n^\dag$ is an operator solution to \eqref{eq:EOMprob_BASIC} it is also an operator solution to
\begin{equation} 
\braket{\Psi_{0}|[\hat{A},[\hat{H},\hat{O}_n^\dag]]|\Psi_{0}} 
= (E_n-E_0)\braket{\Psi_{0}|[\hat{A},\hat{O}_n^\dag]|\Psi_{0}}
\end{equation}
where $\hat{A}$ is an element of $\mathrm{End}(S_\mathbb{K})$.
\noindent We consider the
\begin{equation*}
\hat{\mathcal{T}}_{0\rightarrow n} = \sum_{i=1}^{N}\sum_{a=N+1}^{L} 
  (\bold{x}^n_{ia}\hat{a}^\dag\hat{i}^\wdag 
- \bold{y}^n_{ia}\hat{i}^\dag\hat{a}^\wdag)
\end{equation*}
map. This operator is inspired by the operator introduced in \eqref{eq:def:Tcurved} in which we identify the transition of interest. According to Ref. \cite{lynch_excited_nodate} — see discussion of formula (24) —, substituting $\Psi_0$ by the $\psi_\star$ that is the Hartree-Fock wave function — denoted again $\psi^{\mathrm{HF}}$ below —, and $\hat{O}_n^\dag$ by $\hat{\mathcal{T}}_{0\rightarrow n}$ is the RPA--TDHF approximation, where ``RPA'' stands for ``Random Phase Approximation''. In that case, $\hat{H}$ must be substituted by an $\hatcal{H}$ operator corresponding to the finite-dimensional $\mathcal{B}$ ordered basis, and $\hat{A}$ must be substituted by an $\hatcal{A}$ operator that is an element of $\mathrm{End}(\mathcal{C}_\mathbb{R})$. When choosing $(\hatcal{A} = \hat{i}^\dag\hat{a}$), we get
\begin{equation}\label{eq:EOM-deex-HF}
\left\langle \psi^{\mathrm{HF}} \left|
\left[ \hat{i}^\dag\hat{a}^\wdag, \left[ \hatcal{H}, \hat{\mathcal{T}}_{0\rightarrow n} \right] \right] 
\right|\psi^{\mathrm{HF}} \right\rangle  
= \omega^\mathrm{TDHF}_n
\left\langle \psi^{\mathrm{HF}} \left|
\left[ \hat{i}^\dag\hat{a}^\wdag, \hat{\mathcal{T}}_{0\rightarrow n} \right] 
\right|\psi^{\mathrm{HF}} \right\rangle 
\end{equation}
\noindent In order to recast this equation into a matrix equation, we introduce the $\bold{A}$ and $\bold{B}$ matrices
\begin{align*}
\forall (i,j)\in\llbracket 1,N\rrbracket,\, \forall (a,b)\in\llbracket (N+1),L\rrbracket,\, \bold{A}_{ia,jb} 
&\coloneqq  
\left\langle\psi^{\mathrm{HF}}\left|
\left[ \hat{i}^\dag\hat{a}^\wdag, \left[ \hatcal{H}, \hat{b}^\dag\hat{j}^\wdag \right] \right]
\right|\psi^{\mathrm{HF}} \right\rangle,\\
-\bold{B}_{ia,jb}
&\coloneqq 
\left\langle\psi^{\mathrm{HF}}\left|
\left[ \hat{i}^\dag\hat{a}^\wdag, \left[ \hatcal{H}, \hat{j}^\dag\hat{b}^\wdag \right] \right]
\right|\psi^{\mathrm{HF}} \right\rangle
.
\end{align*}
\noindent We highlight that the definition of the $\bold{B}$ includes a minus sign. Equation \eqref{eq:EOM-deex-HF} reads
\begin{equation*}
\sum_{j=1}^N \sum_{b=N+1}^L (\bold{A}_{ia,jb} \bold{x}_{jb}^n + \bold{B}_{ia,jb} \bold{y}_{jb}^n) = \omega^\mathrm{TDHF}_n\bold{x}_{ia}^n
\end{equation*}
which is equivalent to the matrix equation
\begin{equation}\label{eq:casida_1}
\bold{A}\bold{x}^n + \bold{B}\bold{y}^n = \omega^\mathrm{TDHF}_n \bold{x}^n.
\end{equation}
\noindent If, rather than choosing $(\hatcal{A} = \hat{i}^\dag\hat{a})$ one chooses $(\hatcal{A} = \hat{a}^\dag\hat{i})$, the result is
\begin{equation}\label{eq:casida_2}
-\bold{B}\bold{x}^n - \bold{A}\bold{y}^n = \omega^\mathrm{TDHF}_n \bold{y}^n.
\end{equation}
\noindent Equations \eqref{eq:casida_1} and \eqref{eq:casida_2} can be jointly recast into a generalized eigenvalue problem, namely
\begin{eqnarray}\label{eq:TDHFCASIDA} 
 \left(
  \begin{array}{cc}
    \bold{A}   & \bold{B}    \\
   \bold{B} & \bold{A} \\
  \end{array}\right)
 \left(
  \begin{array}{cc}
   \bold{x}^n \\
   \bold{y}^n \\
  \end{array}\right) = \omega^\mathrm{TDHF}_n
 \left(
  \begin{array}{cc}
    \bold{I}_d   &  \bold{0}_{d} \\
   \bold{0}_{d}  &  -\bold{I}_d   \\
  \end{array}\right)
 \left(
  \begin{array}{cc}
   \bold{x}^n \\
   \bold{y}^n \\
  \end{array}\right)
\end{eqnarray}
with $d=(N\times(L-N))$. We see that $\hatcal{T}_{0\rightarrow n}$ — hence, $\hatcal{T}$ when we consider \textit{any} transition — is regarded as an approximation to the most general form of $\hat{O}^\dag_n$ given in \eqref{eq:def:Odagn}, involved in the search for approximate solutions to \eqref{eq:EOMprob_BASIC}. Due to that interpretation of $\hatcal{T}_{0\rightarrow n}$ we can write the TDHF approximate version of \eqref{eq:lynch24modified}:
$$\braket{\psi^{\mathrm{TDHF}}_n |\hat{A}|\psi_n^{\mathrm{TDHF}}} - \braket{\psi^{\mathrm{HF}} |\hat{A}|\psi^{\mathrm{HF}}} = \braket{\psi^{\mathrm{HF}} |[\hatcal{T}_{0\rightarrow n}^\dag,\hat{A},\hatcal{T}_{0\rightarrow n}^\wdag]|\psi^{\mathrm{HF}}},$$
that is an adapted version of formula (26) in Ref. \cite{lynch_excited_nodate}. We can use Lemma \ref{lemma:jacobi} to show that, for $(\hatcal{A} = \hat{s}^\dag\hat{r})$ with $r$ and $s$ both belonging to $\llbracket 1, L \rrbracket$,
$$\braket{\psi^{\mathrm{HF}} |[\hatcal{T}_{0\rightarrow n}^\dag,\hat{s}^\dag\hat{r},\hatcal{T}_{0\rightarrow n}^\wdag]|\psi^{\mathrm{HF}}} = \braket{\psi^{\mathrm{HF}} |[\hatcal{T}_{0\rightarrow n}^\dag,[\hat{s}^\dag\hat{r},\hatcal{T}_{0\rightarrow n}^\wdag]]|\psi^{\mathrm{HF}}}.$$
Hence, according to \eqref{eq:EOMisFA} the TDHF one-body reduced difference density matrix \textit{is} the $\tilde{\bm{\gamma}}^\Delta$ matrix given in \eqref{eq:RPA_1DDM}. We believe it is important to notice that, though the TDHF one-body reduced difference density matrix derived in the EOM framework coincides with the one given in the original formulation of our problem — i.e., $\tilde{\bm{\gamma}}^\Delta$ —, the two results are obtained by substituting different operators by ``$\hatcal{T}\,$''. Indeed, the substitution operator replaces an operator that has a more general form in the EOM derivations, while in the original formulation of our problem it is $\ket{\Psi_n}\bra{\Psi_0}$ that is substituted by $\hatcal{T}$, though the latter cannot be regarded as an approximation to the former — the $\hatcal{T}$ map is, in the general case, neither nilpotent nor number- or norm-conserving, and in the Tamm-Dancoff limit it is number- and norm-conserving but not nilpotent:
$$((\braket{\psi^\text{HF}|\psi^\text{CIS}} = 0) \land \neg(\hatcal{T}_x\hatcal{T}_x = 0\hat{\mathds{1}}_{\mathcal{C}_\mathbb{R}})) \Longrightarrow \neg (\hatcal{T}_x\ket{\psi^\text{HF}} = \ket{\psi^\text{CIS}} \Longrightarrow \hatcal{T}_x = \ket{\psi^\text{CIS}}\bra{\psi^\text{HF}})$$
where ``$\land$'' (respectively, ``$\neg$'') is here the logical conjunction (respectively, negation) connective. This makes us think that, though the result is identical, the correct framework for deriving the TDHF one-body reduced difference density matrix is the EOM framework. These conclusions actually dissolve the central problem of our paper when TDHF is concerned.

\subsection{The case of TDDFRT}

We have seen in the previous paragraph that the TDHF one-body reduced difference density matrix can be \textit{derived} in the EOM framework because the RPA--TDHF solutions are considered as ``approximate solutions to the EOM'' problem — see again Ref. \cite{lynch_excited_nodate}. However, to our knowledge, this is not the case for TDDFRT. Though the TDDFRT central equation has a structure that is strictly identical to \eqref{eq:TDHFCASIDA}, we are not certain that there exists a universal consensus — or even an actual debate — about whether this equation in the TDDFRT case can be regarded as the result of inserting approximate objects in either \eqref{eq:EOMprob_BASIC} or \eqref{eq:EOMprob_SYMMETRIZED}, i.e., that the TDDFRT solutions can be regarded as approximate solutions to the EOM problem. We believe that it is because the components of \textbf{x} and \textbf{y} in the solutions of \eqref{eq:TDHFCASIDA} usually have the same \textit{interpretation} in TDHF and in TDDFRT that the structure of the one-body reduced difference density matrix is identical for the two methods. From what we know we feel unable to declare that the TDDFRT one-body reduced density matrix is \textit{derived} as an approximate object in the EOM framework, and to dissolve the central problem of our paper with such a conclusion. 

If the TDDFRT one-body reduced density matrix has to be \textit{defined} rather than being \textit{derived}, and if its definition must come from a parallel drawn with TDHF, such a parallel must lie in the TDHF derivation of its one-body reduced difference density matrix in the TDHF--EOM framework rather than in the framework built in the original or symmetrized formulation of the problem in this paper. That is to say, the parallel should not be drawn from the TDHF approximation of \eqref{eq:EOM_1DDM}, but rather a few steps before, i.e., the TDHF approximation of \eqref{eq:lynch24modified}. This is also true when the Bethe-Salpeter Equation is concerned, for the same reasons. 

\section{Conclusions and perspectives}

We have started our discussion with a problem given in a (nested) double-commutator formulation: In the exact case the choice of the operator used for deriving the one-body reduced difference density matrix elements is arbitrary since any choice leads to the same result; in the approximate case different choices of operators — with their structure being identical to those used in the exact case — lead to different matrices. We have proved that only one of them is symmetric, which is a strict requirement. 

\noindent When generalizing our initial problem with a symmetrized (nested) double-commutator formulation, this argument does not hold anymore, because all the possible candidate matrices are symmetric. However, we have seen for that framework that (i) in the TDHF case only one of such matrices can be regarded as a potential candidate because all the others miss the important property that in the Tamm-Dancoff limit they become the required object, and (ii) in the TDDFRT case the same conclusion holds only for strong interpretations of the Tamm-Dancoff solutions.

Recasting our problem in the equations-of-motion framework lead us to make an additional distinction between the TDHF and the TDDFRT cases. The solutions to the former can be regarded as approximate solutions to the EOM problem, and in fact the EOM framework readily provides a procedure for deriving the object that was of interest to us in this study. For that reason the central problem of our paper can be regarded as a non-problem when the TDHF method is concerned. On the other hand we actually miss this interpretation for the TDDFRT solutions — in other words we do not recast the TDDFRT problem into the EOM framework.

This study focused on \textit{one} object in the representation of a molecular electronic transition. More precisely we focused on one property of one object in the representation of a molecular electronic transition — is the candidate one-body reduced difference density matrix symmetric or not? In the text we have also mentioned the fact that such a matrix must have a trace that is equal to zero, but this property was readily fulfilled by all the forms that have been derived, and by all the forms that have been conjectured. A perspective for overcoming the non-universality of the interpretation of the Tamm-Dancoff solutions when the problem is given for the one-body reduced difference density matrix in the symmetrized (nested) double-commutator formulation would be to look at other acceptability criteria for that matrix — e.g., the boundary values for the eigenvalues of the candidate matrices: If those values can be found below $(-1)$ or above $1$, this would question the $N$--representability of the (possibly unknowable) states at stake in the representation. Another perspective would be to look at other objects — e.g., the one-body reduced \textit{transition} density matrix —, or to see if the use of some candidate matrices violates some fundamental \textit{relations}. 

\section*{Acknowledgements}

The authors wish to acknowledge their colleagues from the LPCT research unit for their support and for very fruitful discussions on the topic.

\end{document}